\documentclass[12pt, draftclsnofoot, onecolumn]{IEEEtran}
\usepackage[latin9]{inputenc}
\usepackage{array}
\usepackage{calc}
\usepackage{textcomp}
\usepackage{mathtools}
\usepackage{enumitem}
\usepackage{algorithm2e}
\usepackage{amsmath}
\usepackage{amsthm}
\usepackage{amssymb}
\usepackage{graphicx}
\usepackage{tcolorbox}
\makeatletter

\theoremstyle{plain}

\theoremstyle{plain}

\usepackage{graphicx}
\usepackage{cite}
\usepackage{psfrag}
\usepackage{subfigure}
\usepackage{graphicx}
\usepackage{empheq}
\usepackage{cite}
\usepackage{subfigure}
\usepackage{mathrsfs}
\usepackage[displaymath,mathlines]{lineno}
\usepackage{color}
\usepackage{tabulary}
\usepackage{multirow}
\usepackage{cases}
\usepackage{url}

\usepackage{siunitx}

\usepackage{enumerate}
\usepackage{color}
\usepackage{graphicx}
\usepackage{epsfig}
\usepackage[mathcal]{euscript}

\newtheorem{lemma}{Lemma}

\newtheorem{proposition}{Proposition}

\DeclareMathOperator{\st}{\text{s.t.}}

\DeclareMathOperator{\argmin}{\mathrm{argmin}}

\newcommand{\Ptotal}{P_{\mathrm{total}}}

\newcommand{\Pp}{\rho_{\mathrm{p}}}

\newcommand{\Pd}{\rho_{\mathrm{d}}}
\newcommand{\Se}{{\mathsf{ S}}_{\mathrm{e}}}

\newcommand{\Ee}{{\mathsf{ E}}_{\mathrm{e}}}
\newcommand{\So}{{\mathsf{ S}}_{\mathrm{o}}}

\newcommand{\tauc}{\tau_\mathrm{c}}
\newcommand{\taup}{\tau_\mathrm{p}}

\def\@setsize#1#2#3#4{
    \@nomath#1
    \let\@currsize#1
    \baselineskip #2
    \baselineskip \baselinestretch\baselineskip
    \parskip \baselinestretch\parskip
    \setbox\strutbox \hbox{
        \vrule height.7\baselineskip
            depth.3\baselineskip
            width\z@}
    \skip\footins \baselinestretch\skip\footins
    \normalbaselineskip\baselineskip#3#4}

\newcommand{\setstretch}[1]{
    \def\baselinestretch{#1}%
    \@currsize
    }

\renewcommand{\baselinestretch}{1.35}

\makeatother

\providecommand{\lemmaname}{Lemma}
\providecommand{\theoremname}{Theorem}

\IEEEoverridecommandlockouts

\begin{document}

\title{\huge Energy Efficiency Maximization in Large-Scale Cell-Free Massive MIMO: A Projected Gradient Approach}
\author{Trang C. Mai, Hien Quoc Ngo, and Le-Nam Tran
	\thanks{T. C. Mai and H. Q. Ngo are with the Institute of Electronics, Communications and Information Technology in Queen's University Belfast, Belfast, U.K. (email: \{trang.mai, hien.ngo\}@qub.ac.uk)}
	\thanks{L. N. Tran is with University College Dublin, Ireland (email: nam.tran@ucd.ie)}
	\thanks{The work of Trang C. Mai and Hien Quoc Ngo was supported by the UK Research and Innovation Future Leaders Fellowships under Grant MR/S017666/1. The work of Le-Nam Tran was supported in part by a Grant from Science Foundation Ireland under Grant number 17/CDA/4786.}
	\thanks{Parts of this work were presented at the 2019 Asiloma Conf. \cite{nam2019first}.}
}
\maketitle

\begin{abstract}

This paper considers the fundamental power allocation problem in cell-free massive mutiple-input and multiple-output (MIMO) systems which aims at maximizing the total energy efficiency (EE) under a sum power constraint at each access point (AP) and a quality-of-service (QoS) constraint at each user. Existing solutions for this optimization problem are based on solving a sequence of second-order cone programs (SOCPs), whose computational complexity scales dramatically with the network size. Therefore, they are not implementable for practical large-scale cell-free massive MIMO systems. To tackle this issue, we propose an iterative power control algorithm based on the frame work of an accelerated projected gradient (APG) method. In particular, each iteration of the proposed method is done by simple closed-form expressions, where a  penalty method is applied to bring constraints into the objective in the form of penalty functions. Finally, the convergence of the proposed algorithm is analytically proved and numerically compared to the known solution based on SOCP. Simulations results demonstrate that our proposed power control algorithm can achieve the same EE as the existing SOCPs-based method, but more importantly, its run time is much lower (one to two orders of magnitude reduction in run time, compared to the SOCPs-based approaches).  

\end{abstract}


\section{Introduction\label{Sec:Introduction}}

Cell-free massive mutiple-input and multiple-output (MIMO) has attracted a lot of research interest recently, for its ability to overcome the inherent intercell-interference of cellular networks \cite{NAYLM:16:WCOM,NAMHR:17:WCOM, IBNFL:19:EURASIP_JWCN, zhang2019cell,d2020analysis,attarifar2019modified}. Basically, cell-free massive MIMO uses a massive number of distributed access points (APs) together with simple linear processing to coherently serve many users using the same time and frequency resources. Each APs can be equipped with several antennas. It relies on the favorable propagation and channel hardening property of massive MIMO technology \cite{MLYN:16:Book}, and the macro-diversity of network MIMO technique  \cite{SZ:01:VTC}. Therefore, it can provide universally good service to all users in the network regardless their locations.


In cell-free massive MIMO, because the APs and users are distributed over a large area, power controls are very important to control the near-far effect, and hence, can significantly improve the system performance as well as to save the radiated powers from the APs in the downlink and the users in the uplink. Thus, many research works on power allocations in cell-free massive MIMO have been studied \cite{ NAYLM:16:WCOM, NAMHR:17:WCOM, bashar2019max, qiu2020downlink, van2020joint, Ngo2018a}. In \cite{NAYLM:16:WCOM, NAMHR:17:WCOM}, power control coefficients at the APs and users were optimally chosen to maximize the minimum spectral efficiency of all users. In \cite{bashar2019max}, the max-min power power control under  limited backhaul was investigated.
The downlink transmission power optimization of cell-free massive MIMO with spatially correlated Rayleigh fading channels  for noncoherent joint  and coherent joint transmission was exploited in \cite{qiu2020downlink}. A joint downlink transmit powers and the number of active APs optimization was proposed and solved in \cite{van2020joint}. In \cite{Ngo2018a}, the total energy efficiency maximization taken into account the hardware and backhaul power consumption was proposed and exploited.
In the context of cell-free massive MIMO, the solutions to power control problems in most, if not all, of previous work are based on successive convex approximation principle, which approximates a non-convex problem by a sequence of convex second order cone programs (SCOPs)\cite{alizadeh2003second}. As a result, these methods have very high computational complexity, as they rely on interior
point methods (through the use of off-the-shelf convex solvers) to solve these convex problems. Therefore, they are not implementable for large-scale cell-free massive MIMO with many APs and users (e.g. in stadium or shopping malls where we may have thousands of APs and active users).

To deal with the large-scale problem in cell-free massive MIMO, \cite{bjornson2020scalable} proposed a scalable framework, which uses AP selection to cope with computational complexity and backhaul requirements. However, the power allocation is quite simple and heuristic, and thus, it may underestimate the capacity of the system. Another approach to deal with the large-scale problem is presented in \cite{buzzi2017cell,buzzi2019user,alonzo2019energy}. The main idea of this approach is to  decompose a large optimization problem into  smaller optimization subproblems to reduce computational complexity. However, this method cannot be applied to solve optimization problems in which the  variables are coupled, such as those with quality of service (QoS) constraints.

In this paper, we consider the energy efficiency maximization problem with QoS constraints with an emphasis on large-scale settings. For such scenarios, the number of power
control coefficients (i.e., the product of
the number of APs and the number of users) which need to be optimized can be extremely large. Thus, our goal is to propose a novel power control algorithm for the energy efficiency maximization problem, which has much lower computational complexity compared to the traditional SOCP-based method. It is apparent from the above discussions that an efficient numerical method for this particular problem is still demanding. To this end, we combine the penalty method and the accelerated projected gradient (APG) method. More specifically, the penalty method is used to handle the QoS constraints in our problem, resulting in more tractable subproblems. Note that the penalty method  is widely used to deal with constrained optimization \cite{AkardiOptII}. In principle, the penalty method penalizes a set of constraints by proper terms and adds the penalty terms into the objective, creating the so-called penalized objective. In this way, an optimal problem with sophisticated constraints can be converted into a regularized optimization problem with simple constraints for which efficient solutions are easier to derive. By increasing the penalty parameter, the solutions to these regularized problems converge to a solution of the original problem.  In this paper, to solve the regularized problems obtained from the penalty method, we then apply the APG method which is a variant of the accelerated proximal gradient method proposed in \cite{Li:2015:APG} for nonconvex programming. 
As shall be numerically shown in Section \ref{sec:numresults}, compared to the sequential SOCP-based method, the proposed method achieves the same total EE but with much lower run time and computational complexity since it is entirely based on first order oracle (i.e the value of the objective and its gradient). Thus, the proposed method can be readily modified to tackle the high complexity of other resource allocation problems for large-scale cell-free massive MIMO. The main contributions of this paper are  as follows.

\begin{itemize}
\item We provide the mathematical background of the APG method with a detailed proof as an alternative solution for sequential SOCP-based method to deal with many resource allocation problems in large-scale cell-free massive MIMO.

\item In our proposed APG method, no external optimization solver is needed as the projection, which is the main operation of the proposed method, is done by closed form expression. As the result, it is much faster to output a solution, compared to the known sequential SOCP-based method. 

\item We customize the presented APG method to solve the total EE maximization problem in cell-free massive MIMO.

\item For our specific problem, we first transform the problem of total EE maximization, subject to transmit power constraints at APs and the individual quality-of-service (QoS) constraints at each user, into a form amenable to the application of the APG method. 

\item We then combined the penalty method and the APG method to achieve a good and low-complexity power control algorithm. In particular, the QoS constraints are penalized by a proper smooth penalty term which in controlled by a penalty parameter. The penalty term is then added to the original objective, giving rise to the penalized problem.  The APG is applied to solve the penalized problem, whereby each iteration admits closed-form expressions. The computational complexity of the proposed algorithm is provided.

\item We provide numerical results to show that the proposed algorithm can achieve the same performance as an SCOPs-based method but with much reduced run time.

\item We also verify that our proposed algorithm converges to a feasible solution regardless the choice of starting point.


\end{itemize}
The rest of this paper is organized as follows. Section II provide the preliminaries of the APG method. Next,
Section III recalls total EE optimization problem in cell-free massive MIMO. Then, Section IV
proposes to use APG method for total EE optimization problem. Section V evaluates the system performance by
using numerical results. Finally, the conclusion is drawn in Section VI.

\textit{Notation:} Standard notations are used in this paper. The superscripts $\ensuremath{(\cdot)^{T}}$ and $\ensuremath{(\cdot)^{H}}$ stand for the transpose and the Hermitian, respectively. Notation $[\mathbf{x}]_{+}$ denotes the  projector onto the positive orthant. Notation   $\nabla$, and $\partial$ denote the gradient, and sub-gradient, respectively. Notation $\text{dom}f$ denotes a domain of function $f$, and $||.||$ denotes the $l_2$-norm. Notation $\equiv$ is used to define an equivalent quantity. Finally, 
we use $\odot$, $\otimes$, and $\langle.,.\rangle$ to denote the Hadamard, the Kronecker products and an inner product, respectively.

\section{Mathematical Preliminaries:  Accelerated Projected Gradient Method}
In this section, we first provide the general framework of an accelerated proximal gradient method and then present a variant of the accelerated proximal gradient method, which is termed the accelerated \textit{projected} gradient (APG) method, to deal with the  EE maximization problem to be considered in Section~\ref{sec:teeo}. First, we recall some definitions. A function $f$ is said to be proper if $\text{dom} f \ne 0$. A function $f$ is said to have an $L$-Lipschitz continuous
	gradient if there exists some $L>0$ such that $||\nabla f(\mathbf{x})-\nabla f(\mathbf{y})||\leq L||\mathbf{x}-\mathbf{y}||,\forall\mathbf{x},\mathbf{y}.$
	If $f(x) \ge x_0, ~\forall x \in X$, then the function is said to be bounded from below by $x_0$. If $f(x) \ge x_0, ~\forall x \in X$, then the function is said to be bounded from below by $x_0$. A function $f$ is lower semicontinuous at point $x_0$ if $\liminf _{x \rightarrow \mathbf{x} 0} f(\mathbf{x}) \geq f\left(\mathbf{x}_{0}\right)$. $f(x)$ is coercive, i.e., $f$ is bounded from below and $f(\mathbf{x}) \rightarrow \infty \quad$ when $\quad\|\mathbf{x}\| \rightarrow \infty$.

We now present a general mathematical framework, called the accelerated proximal gradient method for non-convex problems  presented in \cite{Li:2015:APG},  which concerns the following optimization problem
\begin{equation}\label{P3}
\underset{\mathbf{x}\in\mathbb{R}^{n}}{\min}\ \{T(\mathbf{x})\equiv f(\mathbf{x})+g(\mathbf{x})\},
\end{equation}
where $f(\mathbf{x})$ is $L$-Lipschitz continuous
gradient, $g(\mathbf{x})$ is proper and lower semicontinuous, and $T(\mathbf{x})$ is coercive. Then the accelerated proximal gradient method for solving (\ref{P3}), consists of the following iterations:
\begin{equation}\label{eq:APGm}
	\begin{gathered}
	\mathbf{y}_{k}=\mathbf{x}_{k}+\frac{t_{k-1}}{t_{k}}(\mathbf{z}_{k}-\mathbf{x}_{k})+\frac{t_{k-1}-1}{t_{k}}(\mathbf{x}_{k}-\mathbf{x}_{k-1}),\\
	\mathbf{z}_{k+1}=\text{prox}_{\alpha_{y} g}(\mathbf{y}_{k}-\alpha_{y}\nabla f(\mathbf{y}_{k})),\\
	\mathbf{v}_{k+1}=\text{prox}_{\alpha_{x} g}(\mathbf{x}_{k}-\alpha_{x}\nabla f(\mathbf{x}_{k})),\\
	\mathbf{x}_{k+1}=\begin{cases}
	\mathbf{z}_{k+1} & T(\mathbf{z}_{k+1})\le T(\mathbf{v}_{k+1})\\
	\mathbf{v}_{k+1} & \textrm{otherwise},
	\end{cases}\\
	t_{k+1}=\frac{\sqrt{4t_{k}^{2}+1}+1}{2},
	\end{gathered}
\end{equation}
where $\alpha_{x}$, $\alpha_{y}$ are step sizes, and $\text{prox}_{\alpha g}$ is the proximal operator defined as 
\begin{equation}\label{pro_op1}
\text{prox}_{\alpha g}(\mathbf{x})\triangleq \underset{\mathbf{u}}{\argmin}\ g(\mathbf{u})+\frac{1}{2\alpha}||\mathbf{x}-\mathbf{u}||^{2}.
\end{equation}

In general, accelerated proximal gradient method in (\ref{eq:APGm}) is designed to cope with the \textit{unconstrained optimization} problem (\ref{P3}). However, most of resource allocation problems in cell-free massive MIMO are constrained optimization problems. Therefore, in this paper, we present a special case  of the accelerated proximal gradient method, which is called the APG method, to deal with those problems. Specifically, we consider the following optimization problem
\begin{equation}\label{P4}
\underset{\mathbf{x}\in\mathcal{C}}{\min}\ f(\mathbf{x}),
\end{equation}
where $\mathcal{C}$ is the feasible set of the considered problem, which is often defined by a set of constraints. Again, assume that $f(\mathbf{x})$ is a proper function with Lipschitz continuous
gradient, and bounded from below. We remark that $g(\mathbf{x})$ in \eqref{P3} is not necessarily smooth. Thus, to apply the iterations in (\ref{eq:APGm}) to solve (\ref{P4}), we can let $g(\mathbf{x})$ in (\ref{P3}) be the indicator function of the feasible set $\mathcal{C}$. In this way, the proximal operator in (3) reduces to the Euclidean projection onto $\mathcal{C}$ \cite{Li:2015:APG}. As a result, the APG method for solving \eqref{P4}, consists
of the following iterations:
\begin{tcolorbox}[boxrule=0.5pt,colback = white,sharp corners,before skip=1pt, after skip=4pt,left=0mm,right=0mm,bottom=0pt]
\vspace{-\abovedisplayskip}		
\vspace{-\abovedisplayskip}
\begin{subequations}\label{eq:mAPG}
				\begin{gather}
				\mathbf{y}_{k}=\mathbf{x}_{k}+\frac{t_{k-1}}{t_{k}}(\mathbf{z}_{k}-\mathbf{x}_{k})+\frac{t_{k-1}-1}{t_{k}}(\mathbf{x}_{k}-\mathbf{x}_{k-1})\\
				\mathbf{z}_{k+1}=P_{\mathcal{C}}(\mathbf{y}_{k}-\alpha_{y}\nabla f(\mathbf{y}_{k}))\\
				\mathbf{v}_{k+1}=P_{\mathcal{C}}(\mathbf{x}_{k}-\alpha_{x}\nabla f(\mathbf{x}_{k}))\\
				\mathbf{x}_{k+1}=\begin{cases}
				\mathbf{z}_{k+1} & f(\mathbf{z}_{k+1})\le f(\mathbf{v}_{k+1}),\\
				\mathbf{v}_{k+1} & \textrm{otherwise}
				\end{cases}\\
				t_{k+1}=\frac{\sqrt{4t_{k}^{2}+1}+1}{2},
				\end{gather}
			\end{subequations} %
\end{tcolorbox}
\noindent where  $P_{\mathcal{C}}(\mathbf{x})$ denotes the Euclidean projection of $\mathbf{x}$ onto $\mathcal{C}$, which is defined as
\begin{equation}\label{pro_op2}
P_{\mathcal{C}}(\mathbf{x})\triangleq \underset{\mathbf{u} \in \mathcal{C}}{\argmin}||\mathbf{x}-\mathbf{u}||^{2}.
\end{equation}
\section{Total Energy Efficiency Optimization in Cell-Free Massive MIMO}\label{sec:teeo}
In this section,  we first briefly introduce the system model of cell-free massive MIMO and formulate the total EE optimization problem with conjugate beamforming at the APs, taking into account arbitrary pilot sequence assignments, and imperfect channel estimation. Then, by using penalty functions (PFs), we reformulate the total EE optimization problem into the form so that the APG method can be applied.
\subsection{System Model}
\begin{figure}[h!]
	\centering
	\includegraphics[width=0.65\columnwidth]{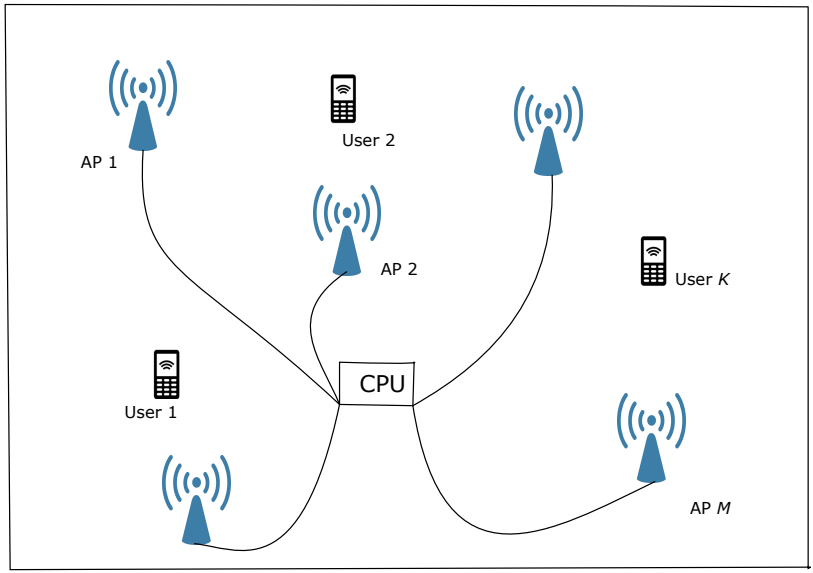}
	
	\caption{System Model.}
	\label{fig:Sysmol}
\end{figure}
\subsubsection{Spectral Eficiency}
We consider a cell-free massive MIMO downlink, which is shown in Fig. \ref{fig:Sysmol}, where $M$ APs coherently serve $K$ users. To take the advantage of channel reciprocity, we consider time division duplex (TDD) operation, where channel information only needs to be estimated in the uplink training phase, and is used in both uplink and downlink data transmission phases. All $M$ APs connect to a central processing unit (CPU) through
a backhaul network. Each user has a single antenna, while each AP is equipped with $N$ antennas. The propagation channel between  AP $m$ and user $k$ is modeled as
 \begin{equation}
 \mathbf{g}_{m k}=\beta_{m k}^{1 / 2} \mathbf{h}_{m k},
 \end{equation}
 where $\beta_{m k}$ is the large-scale fading, and $\mathbf{h}_{m k}$ is the small-scale fading, whose elements are  i.i.d. $\mathcal{CN}(0,1)$ RVs.
The downlink transmission needs two phases: uplink training and downlink payload data transmission phases. In the training phase, all users  send their pilot sequences, $\sqrt{\tau_{\mathrm{p}}} \boldsymbol{\varphi}_{k} \in \mathbb{C}^{\tau_{\mathrm{p}} \times 1},~\forall k$, where $\left\|\boldsymbol{\varphi}_{k}\right\|^{2}=1$, to all APs in the system. Then,  pilot signal received at AP $m$ is
\begin{equation}
\mathbf{Y}_{\mathrm{p}, m}=\sqrt{\tau_{\mathrm{p}} \rho_{\mathrm{p}}} \sum_{k=1}^{K} \mathbf{g}_{m k} \boldsymbol{\varphi}_{k}^{H}+\mathbf{W}_{\mathrm{p}, m},
\end{equation}
where $\rho_{\mathrm{p}}$ is the normalized transmit signal-to-noise ratio
(SNR) of each pilot symbol, $\mathbf{W}_{\mathrm{p}, m}$  is the noise matrix whose elements are i.i.d. $\mathcal{CN}(0,1)$ RVs. After that, each AP uses its received pilot signals from all $K$ users to estimate its local channels using the minimum mean-square error (MMSE) technique \cite{KAY:93:Book}. The channel estimate of ${\mathbf{g}}_{m k}$ is
\begin{equation}
\hat{\mathbf{g}}_{m k}=\frac{\sqrt{\tau_{\mathrm{p}} \rho_{\mathrm{p}}} \beta_{m k}}{\tau_{\mathrm{p}} \rho_{\mathrm{p}} \sum_{k^{\prime}=1}^{K} \beta_{m k^{\prime}}\left|\varphi_{k^{\prime}}^{H} \varphi_{k}\right|^{2}+1}{\check{\mathbf{y}}}_{\mathrm{p}, m k}, 
\end{equation}
where 
\begin{equation}
\begin{aligned}
\check{\mathbf{y}}_{\mathbf{p}, m k}  \triangleq \sqrt{\tau_{\mathrm{p}} \rho_{\mathrm{p}}} \mathbf{g}_{m k}+\sqrt{\tau_{\mathrm{p}} \rho_{\mathrm{p}}} \sum_{k^{\prime} \neq k}^{K} \mathbf{g}_{m k^{\prime}} \boldsymbol{\varphi}_{k^{\prime}}^{H} \boldsymbol{\varphi}_{k}+\mathbf{W}_{\mathrm{p}, m} \varphi_{k}.
\end{aligned}
\end{equation}
In the downlink data transmission phase, APs use conjugate beamforming technique and the channels estimated in the training phase to precode the symbols intended for all users. The vector of transmitted signals from  AP $m$ is
\begin{equation}
\mathbf{x}_{m}=\sqrt{\rho_{\mathrm{d}}} \sum_{k=1}^{K} \sqrt{\eta_{m k}} \hat{\mathbf{g}}_{m k}^{*} q_{k},
\end{equation}
where $q_{k}$ is the symbol intended for user $k$ which satisfies $\mathbb{E}\left\{\left|q_{k}\right|^{2}\right\}=1, ~\forall k$, $\rho_{\mathrm{d}}$ is the normalized transmit power at each AP, and $\eta_{m k}$ is power coefficient between the AP $m$ and the user $k$, satisfying
\begin{equation}
\sum_{k=1}^{K} \eta_{m k} \gamma_{m k} \leq \frac{1}{N},~\forall m,
\end{equation}
where
\begin{align}
\gamma_{mk}  =\frac{{\taup\Pp}\beta_{mk}^{2}}{\taup\Pp\sum_{k'=1}^{K}\beta_{mk'}\left|\pmb{\varphi}_{k'}^{H}\pmb{\varphi}_{k}\right|^{2}+1}.
\end{align}
Then the  signal received at user $k$ is given by
\begin{equation}
r_{k}=\sum_{m=1}^{M} \mathbf{g}_{m k}^{T} \mathbf{x}_{m}+w_{k}.
\end{equation}
By applying the same technique as the one in \cite{Ngo2018a}, the downlink achievable spectral efficiency (SE) of the $k$-th user can be expressed as
\begin{equation}
{\Se}_{k}\left(\{\eta_{mk}\}\right)=\frac{\tauc-\taup}{\tauc}\log_{2}\left(1+\frac{\Pd N^{2}|\bar{\pmb{\gamma}}_{kk}^{T}\bar{\pmb{\eta}}_{k}|^{2}}{\Pd N^{2}\sum\limits _{k'\neq k}^{K}\!|\bar{\pmb{\gamma}}_{k'k}^{T}\bar{\pmb{\eta}}_{k'}|^{2}+\Pd N\sum\limits _{k'=1}^{K}||\pmb{\kappa}_{k'k}\odot\bar{\pmb{\eta}}_{k'}||_{2}^{2}+1}\right)\label{eq:Theo_rateexpr},
\end{equation}
where $\tauc$  is the
length of each coherence interval, $\taup$ is the length of the training phase for each coherence interval, $\pmb{\kappa}_{k'k}\triangleq [\sqrt{\gamma_{1k'}\beta_{1k}};\sqrt{\gamma_{2k'}\beta_{2k}};\ldots;\sqrt{\gamma_{Mk'}\beta_{Mk}}]\in\mathbb{R}_{+}^{M}$,
 $\bar{\pmb{\eta}}_{k}\triangleq[\sqrt{{\eta}_{1k}},\ldots,\sqrt{{\eta}_{Mk}}]^{T}\in\mathbb{R}_{+}^{M}$, and
$\bar{\pmb{\gamma}}_{k'k}\triangleq|\pmb{\varphi}_{k'}^{H}\pmb{\varphi}_{k}|\left[\gamma_{1k'}\frac{\beta_{1k}}{\beta_{1k'}},\gamma_{2k'}\frac{\beta_{2k}}{\beta_{2k'}},\ldots,\gamma_{Mk'}\frac{\beta_{Mk}}{\beta_{Mk'}}\right]^{T}$.  
\subsubsection{Power Consumption Model}
In this paper, the total power consumption is modeled as \cite{Ngo2018a} 
\begin{equation} 
P_{\mathrm {total}}= \sum _{m=1}^{M} P_{m} + \sum _{m=1}^{M} P_{\text {bh},m},
\end{equation}
where $P_{m}$ is the power consumption at the $m$-th AP, and $P_{\text {bh},m}$ is the power consumed by the backhaul link connecting the CPU  and the $m$-th AP. Specifically, $P_{m}$ is modeled as 
\begin{equation} 
P_{m} = \frac {1}{\alpha _{m}} \rho _{\mathrm {d}} N_{0} \left ({N \sum _{k=1}^{K} \eta _{mk}\gamma _{mk}}\right ) + N P_{\text {tc},m}, 
\end{equation}
where $0\le \alpha _{m} \le 1$ is the power amplifier efficiency, $N_{0}$ is the noise power, and $P_{\text {tc},m}$ is the internal power required to run the circuit components at each antenna of the $m$-th AP. Next, $P_{\text {bh},m}$ is modeled as
\begin{equation} 
P_{\text {bh},m} = P_{0,m} +B \cdot {\mathsf { S}}_{\mathrm {e}} \left ({\{\eta _{mk}\}}\right ) \cdot P_{\text {bt},m},
\end{equation}
where $P_{0,m}$ is a fixed power consumption of each backhaul, $P_{\text {bt},m}$ is the traffic-dependent power, and $B$ is the system bandwidth. 
\subsubsection{Total Energy Efficiency\label{Sec:SysModel} }
While spectral efficiency has been a common performance measure for wireless communication design, in this paper we aim to maximize the total energy efficiency of the system, which is defined as how many bits can be transmitted by one Joule. Specifically, the total EE (bit/Joule) can be calculated as 
\begin{align}
\Ee\left(\{\eta_{mk}\}\right)=\frac{{B\sum_{k=1}^{K}{\Se}_{k}(\{\eta_{mk}\})}}{\Ptotal},\label{eq:EE1}
\end{align}

\subsection{Optimization Problem Formulation}
Our problem is to maximize  the total EE \eqref{eq:EE1} by allocating the power coefficients $\{\eta_{mk}\}$, under a sum power constraint at each AP and a QoS constraint, i.e., SE constraint ${\So}_{k}$ at each user. The optimization problem is stated as
\begin{align}
(\mathcal{P}):\left\{ \begin{array}{ll}
\mathop{\max}\limits _{\{\eta_{mk}\}} & {\Ee}(\{\eta_{mk}\})\\
\st%
& {\Se}_{k}(\{\eta_{mk}\}) \ge {\So}_{k}, ~\forall k,\\
& \sum_{k=1}^{K}\eta_{mk}\gamma_{mk}\leq1/N,~\forall m,\\
& \eta_{mk}\geq0,~\forall k,~\forall m,
\end{array}\right.\label{eq:opt1}
\end{align}
In this paper $(\mathcal{P})$ is assumed to be feasible. An equivalent form of problem $(\mathcal{P})$  can be rewritten as 
\begin{align}
(\mathcal{P}_1):\left\{ \begin{array}{ll}
\mathop{\max}\limits _{\{\eta_{mk}\}} & \frac{{B\sum_{k=1}^{K}{\Se}_{k}(\{\eta_{mk}\})}}{\bar{P}_{\mathrm{fix}}+\rho_{\mathrm{d}} N_{0} N \sum_{m=1}^{M} \frac{1}{\alpha_{m}} \sum_{k=1}^{K} \eta_{m k} \gamma_{m k}}\\
\st%
& {\Se}_{k}(\{\eta_{mk}\}) \ge {\So}_{k}, ~\forall k,\\
& \sum_{k=1}^{K}\eta_{mk}\gamma_{mk}\leq1/N,~\forall m,\\
& \eta_{mk}\geq0,~\forall k,~\forall m,
\end{array}\right.\label{eq:opt2}
\end{align}
where $\bar{P}_{\mathrm{fix}} \triangleq \sum_{m=1}^{M}\!\!\left(NP_{\text{tc},m}\!+\!P_{0,m}\right)$. 
Note that the objective function of the problem $(\mathcal{P}_1)$ is nonconvex. The common method to tackle such such a nonconvex problem is to iteratively approximate a nonconvex function by a series of convex functions under the framework of successive convex approximation. In fact, this is the method presented in \cite{Ngo2018a}, in which the EE maximization problem is solved by a sequence of SOCPs. However, as the complexity dramatically increases when the system scales up (i.e. the numbers of APs and users increase), 
such method cannot provide a solution for 
large-scale optimization problems in cell-free massive MIMO systems with thousands of APs and  users. In the next section, we will propose a new algorithm based on the APG method to solve problem $(\mathcal{P}_1)$. Our proposed algorithm has very low complexity, and hence, and can efficiently deal with the systems with many APs and users.

%

\section{Proposed APG Algorithm for the Total Energy Efficiency Optimization}
In this section, we apply the APG method to efficiently solve the total energy efficiency optimization problem in Section \ref{sec:teeo}. We first reformulate the optimization problem $(\mathcal{P}_1)$  by change of variables such that the gradient of the function and the resulting projection can be computed more efficiently. We also apply a penalty method to convert the total EE maximization problem into the form which is amenable to applying the APG method. Note that our algorithm is done over large-scale fading time scale and is performed at the CPU. The details of these steps are described in the following.
\subsection{Problem Reformulation}
%
%
%
%
%
In order to apply the APG method for solving problem ($\mathcal{P}_1$), we need to reformulate ($\mathcal{P}_1$) into the form of (\ref{P4}) and to make sure that the gradient of the objective is Lipschitz continuous. To this end, we first introduce a new variable
${\theta}_{mk}=\sqrt{\eta_{mk}\gamma_{mk}}$ and
define new notations as follows:
\begin{itemize}
	\item $\pmb{\theta}\triangleq(\pmb{\theta}_1;\pmb{\theta}_2;.\ldots;\pmb{\theta}_M)\in\mathbb{R}_{+}^{MK}$, where $\pmb{\theta}_{m}\triangleq[\theta_{m1};\ldots;\theta_{mK}]\in\mathbb{R}_{+}^{K}$, is the vector of all power control coefficients associated with AP
	$m$.
	\item $\mathbf{A}_{k} \triangleq \mathbf{I}_{M}\otimes\mathbf{e}_{k}^{T}$,
	where $\mathbf{e}_{k}\in\mathbb{R}^{K}$ denotes the $k$-th unit
	vector, i.e., the vector such that $e_{k}=1$ and $e_{j}=0,\forall j\neq k$.
	\item $\tilde{\pmb{\gamma}}_{k'k}\triangleq|\pmb{\varphi}_{k'}^{H}\pmb{\varphi}_{k}|\left[\sqrt{\gamma_{1k'}}\frac{\beta_{1k}}{\beta_{1k'}};\negthickspace\sqrt{\gamma_{2k'}}\frac{\beta_{2k}}{\beta_{2k'}};\negthickspace\ldots;\negthickspace\sqrt{\gamma_{Mk'}}\frac{\beta_{Mk}}{\beta_{Mk'}}\right],\notag
	$
	and $\tilde{\pmb{\kappa}}_{k}\triangleq [\sqrt{\beta_{1k}};\sqrt{\beta_{2k}};\ldots;\sqrt{\beta_{Mk}}]\in\mathbb{R}_{+}^{M}$. 
\end{itemize}
Then, ($\mathcal{P}_1$) can be rewritten as 
\begin{align}
(\mathcal{P}_2):\left\{ \begin{array}{ll}
\mathop{\max}\limits _{\pmb{\theta}\in \mathcal{C}} & B\frac{u(\pmb{\theta})}{v(\pmb{\theta})}\triangleq f(\boldsymbol{\theta})\\
\st%
& u_{k}(\pmb{\theta}) \ge {\So}_{k}, ~\forall k,
\end{array}\right.\label{eq:opt3}
\end{align}
where 
\begin{equation}
\mathcal{C}=\{\pmb{\theta}\ |\ ||\pmb{\theta}_{m}||^{2}\leq\frac{1}{N},m=1,2,\ldots,M;\pmb{\theta}\geq 0\},
\end{equation}
\begin{equation}\label{eq:rate:rewrite2}
v(\pmb{\theta})=\bar{P}_{\mathrm{fix}}+\Pd N_{0}N\sum\limits _{m=1}^{M}\frac{1}{\alpha_{m}}||\pmb{\theta}_{m}||^{2},
\end{equation}
\begin{equation}
u(\pmb{\theta})  \triangleq \sum_{k=1}^{K}u_{k}(\pmb{\theta}),
\end{equation}
and
\begin{align}
u_{k}(\pmb{\theta}) & =\frac{\tauc-\taup}{\tauc}\log_2\Biggl(1+\frac{\Pd N^{2}\bigl(\tilde{\pmb{\gamma}}_{kk}^{T}\mathbf{A}_{k}\pmb{\theta}\bigr)^{2}}{\Pd N^{2}\sum\limits _{k'\neq k}^{K}\bigl(\tilde{\pmb{\gamma}}_{k'k}^{T}\mathbf{A}_{k'}\pmb{\theta}\bigr)^{2}+\Pd N\sum\limits _{k'=1}^{K}||\tilde{\pmb{\kappa}}_{k}\odot(\mathbf{A}_{k'}\pmb{\theta})||^{2}+1}\Biggr).\label{eq:rate:rewrite}
\end{align}
As shall be seen shortly, the projection onto $\mathcal{C}$ can be done by closed form expressions. Thus, the main obstacle in deriving an efficient algorithm for solving  $({\mathcal{P}_2})$ is the QoS constraints. To overcome this issue we propose to combine the penalty method and the APG method as described in the next subsection.
\subsection{Proposed Algorithm}
The overall structure of the proposed method is as follows:
\begin{itemize}
    \item The penalty method is invoked to bring the QoS constraints into the objective by some form of a loss function through a penalty parameter, leading to the penalized problem.
    \item The APG method is then applied to solve the penalized problem. This process is repeated until a stopping criterion is achieved.
\end{itemize}

\subsubsection{Penalty Method}  The constraint 
$ u_{k}(\pmb{\theta}) \ge {\So}_{k}$ can be written as
\begin{equation}
	\tilde{\pmb{\gamma}}_{kk}^{T}\mathbf{A}_{k}\pmb{\theta}\geq a_k \sqrt{\Pd N^{2}\sum\limits _{k'\neq k}^{K}\bigl(\tilde{\pmb{\gamma}}_{k'k}^{T}\mathbf{A}_{k'}\pmb{\theta}\bigr)^{2}+\Pd N\sum\nolimits _{k'=1}^{K}||\tilde{\pmb{\kappa}}_{k}\odot(\mathbf{A}_{k'}\pmb{\theta})||^{2}+1}, 
\end{equation}
where $a_k \triangleq\sqrt{\frac{2^{{\So}_{k}\frac{\tauc}{\tauc-\taup}} -1}{\Pd N^{2}}}$. Then, for each QoS constraint we introduce the following quadratic loss function \cite{AkardiOptII}
 \begin{equation}
 \Psi_{k} (\pmb{\theta}) \triangleq \left[\text{max}\left(0,  { g}_{k}(\pmb{\theta})\right )\right]^2,
 \end{equation}
where
 \begin{equation}
 g_k(\pmb{\theta}) \triangleq a_k \sqrt{\Pd N^{2}\sum\nolimits _{k'\neq k}^{K}\bigl(\tilde{\pmb{\gamma}}_{k'k}^{T}\mathbf{A}_{k'}\pmb{\theta}\bigr)^{2}+\Pd N\sum\nolimits _{k'=1}^{K}||\tilde{\pmb{\kappa}}_{k}\odot(\mathbf{A}_{k'}\pmb{\theta})||^{2}+1} -\tilde{\pmb{\gamma}}_{kk}^{T}\mathbf{A}_{k}\pmb{\theta}.
 \end{equation}
Note that  $g_k(\pmb{\theta})$ is convex and $\Psi_{k} (\pmb{\theta})$ is smooth. Then, for a given penalty coefficient $\xi$, the  penalized objective function of   $({\mathcal{P}_2})$, denoted by  $f_{\xi}(\pmb{\theta})$, is given by 
\begin{equation}\label{eq:f_theta}
f_{\xi}(\pmb{\theta})\triangleq B\frac{u(\pmb{\theta})}{v(\pmb{\theta})} -  \xi \sum_k^K\Psi_{k}(\pmb{\theta}).
\end{equation}
We remark that the above regularized objective is formed in the context of maximization. Also note that the value of the penalty coefficient $\xi$ should be selected appropriately. If this parameter is large, the feasibility is guaranteed  but
the resulting optimization problem is numerically ill-conditioned. On the other hand, if it is too small, it may produce a suboptimal solution or even converges to an infeasible point, i.e. the constraints are violated. 

In practice, to avoid the above issues, we can first solve the penalized optimization problem for a small value of $\xi$ and check if the stopping criterion is met. If not, we can increase  $\xi$ by $\rho >1$ times and repeat this process until the stopping criterion is met. In this iterative process, it is critical to use the solution of the previous iteration as the starting point of the next. In essence, the key to  the penalty method  is  to solve the following regularized optimization problem for a given $\xi$
{\color{black}\begin{equation}\label{eq:opt_theta}
\underset{\pmb{\theta}\in \mathcal{C}}{\max}\quad f_{\xi}(\pmb{\theta}),
\end{equation}
which has the same form as (\ref{P4})}. We are now in a position to apply the APG method to solve \eqref{eq:opt_theta} which is detailed next section.

\subsubsection{APG Method}
We first show that  $f_{\xi}(\pmb{\theta})$  is a proper function with Lipschitz continuous gradient and bounded from above, and thus, the APG method is applicable to solve (\ref{eq:opt_theta}).
Towards this end, It is easy to see that the function $f(\pmb{\theta})$ is proper and bounded from above\footnote{Note that for a minimization problem the objective should be bounded from below.}, which is shown by the following inequalities
\begin{align}
f_{\xi}(\pmb{\theta})=B\frac{\tauc-\taup}{\tauc}\frac{u(\pmb{\theta})}{v(\pmb{\theta})} -  \xi \sum_k^K\left[\text{max}\left(0, g_{k}(\pmb{\theta}) \right )\right]^2\le B\frac{\tauc-\taup}{\tauc}\frac{u(\pmb{\theta})}{v(\pmb{\theta})} < \infty. \notag
\end{align}
The above inequality holds since the total EE is bounded from above as its numerator, i.e. total SE, is limited by inter-user interference and total power consumption. The Lipschitz continuity of the gradient of $f_{\xi}(\pmb{\theta})$ is stated in Proposition \ref{pp1}.\begin{proposition}\label{pp1}
	The function $f_{\xi}(\pmb{\theta})$ shown in (\ref{eq:f_theta}) is $L_{f}$-Lipschitz continuous
	gradient with a Lipschitz contant $L_{f}$ given by (\ref{eq_lf1}) in Appendix \ref{sec_lips_proof}.
\end{proposition}
\begin{proof}
	See Appendix \ref{sec_lips_proof}.
\end{proof}

It is now obvious that we can apply the APG method in (\ref{eq:mAPG}) to solve (\ref{eq:opt_theta}). Our proposed method that combines the penalty method and the APG method is summarized in Algorithm \ref{alg:APG}, where  $\boldsymbol{\theta}_{\xi}$
denote an optimal solution to (\ref{eq:opt_theta}). 
\begin{algorithm}[!h]
	\hrulefill{}
	\caption{The proposed algorithm for solving \eqref{eq:opt3}.\label{alg:APG}}
	\SetAlgoNoLine
	\SetNlSty{textnormal}{}{}
	
	\KwIn{ $\pmb{\theta}^{(0)}\in\mathbb{R}_{+}^{MK}$,  $0 < \alpha_{\theta},\alpha_{y} < 1/L_{f}$, $\rho>1$, $\delta>0$, $\varsigma > 0$, $\xi$} 
\SetKwInput{KwIn}{Initialization}
	\KwIn{$\pmb{\theta}^{(1)}=\mathbf{z}^{(1)}=\pmb{\theta}^{(0)}$}
	\Repeat(\tcc*[f]{outer loop: penalty method}){convergence}{	
	Set $t^{(1)}=t^{(0)}=1$; $n\leftarrow 1$; $m\leftarrow1$\\
\Repeat(\tcc*[f]{inner loop: APG method}){$\left|\frac{f_{\xi_{m} }(\boldsymbol{\theta}^{(n)})-f_{\xi_{m} }(\boldsymbol{\theta}^{(n-10)})}{f_{\xi_{m} }(\boldsymbol{\theta}^{(n)})}\right|\leq \varsigma$}{
		\begin{subequations}\label{eq:stepsize:project}
		\begin{align}
		\mathbf{y}^{(n)}&=\pmb{\theta}^{(n)}+\frac{t^{(n-1)}}{t^{(n)}}(\mathbf{z}^{(n)}-\pmb{\theta}^{(n)}) +\frac{t^{(n-1)}-1}{t^{(n)}}(\pmb{\theta}^{(n)}-\pmb{\theta}^{(n-1)})\notag\\
		\mathbf{z}^{(n+1)}&=P_{\mathcal{C}}(\mathbf{y}^{(n)}+\alpha_{y}\nabla f_{\xi_{m}}(\mathbf{y}^{(n)})) \label{eq:updatez}\\
		\mathbf{v}^{(n+1)}&=P_{\mathcal{C}}(\pmb{\theta}^{(n)}+\alpha_{\theta}\nabla f_{\xi_{m}}(\pmb{\theta}^{(n)})) \label{eq:updatev}\\
		\pmb{\theta}^{(n+1)}&=\begin{cases}
		\mathbf{z}^{(n+1)}, \quad\text{if} \  f_{\xi_{m} }(\mathbf{z}^{(n+1)})\ge f_{\xi_{m} }(\mathbf{v}^{(n+1)})\\
		\mathbf{v}^{(n+1)}, \quad \textrm{otherwise},
		\end{cases}\notag\\
		t^{(n+1)}&=\frac{\sqrt{4(t^{(n)})^{2}+1}+1}{2}\notag\\
		n&\leftarrow n+1 \nonumber
		\end{align}			
		\end{subequations}
	}
Update the starting point for the next iteration:	$\pmb{\theta}^{(1)}=\mathbf{z}^{(1)}=\pmb{\theta}^{(n)}$\\
Set $\pmb{\theta}_{\xi_{m}}=\pmb{\theta}^{(n)}$\\
Increase the penaly parameter:  $\xi_{m+1} =\xi_{m} \times \rho $	\\
$m\leftarrow m+1$
}

	\hrulefill{} \\ ~	
\end{algorithm}
Regarding the APG procedure in Algorithm \ref{alg:APG}, we note that we have modified (\ref{eq:mAPG}), accounting for the maximization context, where we move along the gradient  to increase the objective of the current point. Note also that, for a practical purpose we stop the APG procedure when the relative increase in the objective during the last 10 iterations is less than $\varsigma$.

It is clear that the key operations in the implementation of Algorithm \ref{alg:APG}  are the computation of the gradient  $\nabla f_{\xi}(\pmb{\theta})$ and the projections in \eqref{eq:updatez} and \eqref{eq:updatev}. In particular,  these two operations can be done in closed-form as shown in Proposition \ref{lem:grad} and Proposition \ref{lem:proj}, respectively.
\begin{proposition}\label{lem:grad}
	$\nabla f_{\xi}(\pmb{\theta})$ can be calculated as
	\begin{align}
	\nabla f_{\xi}(\pmb{\theta})= B\frac{v(\pmb{\theta})\nabla u(\pmb{\theta})-u(\pmb{\theta})\nabla v(\pmb{\theta})}{v(\pmb{\theta})^{2}} - \xi\sum_{k=1}^{K} \nabla \Psi_{k} (\pmb{\theta}),
	\end{align}
	where 
	\begin{equation}
	\nabla u(\pmb{\theta})= \sum_{k=1}^{K} \nabla u_{k}(\pmb{\theta}),
	\end{equation}
	\begin{equation}
	\nabla v(\pmb{\theta})=\Pd N_{0}N\left[\frac{2}{\alpha_{1}}\pmb{\theta}_{1};\frac{2}{\alpha_{2}}\pmb{\theta}_{2};\ldots,\frac{2}{\alpha_{M}}\pmb{\theta}_{M}\right],
	\end{equation}
\begin{equation}
	\nabla \Psi_{k} (\pmb{\theta}) = 2  \left[\max\left(0, { g}_{k}(\pmb{\theta}) \right )\right] \nabla g_{k}(\pmb{\theta}),
	\end{equation}
\begin{equation}\label{eq_gk}
\nabla g_k(\pmb{\theta}) =  {N a_{k}\Pd \sqrt{\bar{\mu}_{k}}\Bigl(\sum_{k'\ne k}^{K}\mathbf{A}_{k'}^{T}\bigl(\tilde{\pmb{\gamma}}_{k'k}\tilde{\pmb{\gamma}}_{k'k}^{T}+\frac{1}{N}\mathbf{B}_{k}\bigr)\mathbf{A}_{k'}+\frac{1}{N}\mathbf{A}_{k}^{T}\mathbf{B}_{k}\mathbf{A}_{k}\Bigr)\pmb{\theta}} -\mathbf{A}_{k}^{T}\tilde{\pmb{\gamma}}_{kk},
\end{equation}
	and
			\begin{align}\label{eq:grad:SE}
			\nabla u_k(\pmb{\theta}) & =\frac{\Pd\mu_{k}(\tauc-\taup)}{\tauc\ln{2}}\Bigr(\mathbf{A}_{k}^{T}\tilde{\pmb{\gamma}}_{kk}\tilde{\pmb{\gamma}}_{kk}^{T}\mathbf{A}_{k}-\bar{\mu}_{k}\Pd \bigl(\tilde{\pmb{\gamma}}_{kk}^{T}\mathbf{A}_{k}\pmb{\theta}\bigr)^{2}\notag\\
			&\quad\times\bigl(\sum\nolimits_{k'\neq k}^{K}\mathbf{A}_{k'}^{T}\bigl(\tilde{\pmb{\gamma}}_{k'k}\tilde{\pmb{\gamma}}_{k'k}^{T}+\frac{1}{N}\mathbf{B}_{k}\bigr)\mathbf{A}_{k'}+\frac{1}{N}\mathbf{A}_{k}^{T}\mathbf{B}_{k}\mathbf{A}_{k}\bigr) \Bigr)\pmb{\theta},
			\end{align}
where
			\begin{align}
			\mu_{k} & \triangleq \frac{2}{\Pd\sum_{k'=1}^{K}\bigl(\tilde{\pmb{\gamma}}_{k'k}^{T}\mathbf{A}_{k'}\pmb{\theta}\bigr)^{2}+\frac{\Pd}{N}\sum_{k'=1}^{K}||\tilde{\pmb{\kappa}}_{k}\odot(\mathbf{A}_{k'}\pmb{\theta})||^{2}+\frac{1}{N^{2}}},\notag\\
			\bar{\mu}_{k} & \triangleq \frac{1}{\Pd\sum_{k'=1,k'\neq k}^{K}\bigl(\tilde{\pmb{\gamma}}_{k'k}^{T}\mathbf{A}_{k'}\pmb{\theta}\bigr)^{2}+\frac{\Pd}{N}\sum_{k'=1}^{K}||\tilde{\pmb{\kappa}}_{k}\odot(\mathbf{A}_{k'}\pmb{\theta})||^{2}+\frac{1}{N^{2}}}.\notag
			\end{align}

\end{proposition}
\begin{IEEEproof}
See Appendix \ref{sec:grad_proof}.
\end{IEEEproof}
%
\begin{proposition}\label{lem:proj}
	The projection $P_{\mathcal{C}}(\mathbf{u})$ admits the following analytical solution
	\begin{equation}\label{eq:anal_sol}
	\pmb{\theta}_{m}=\frac{\sqrt{1/N}}{\max(||\bigl[\mathbf{u}_{m}\bigr]_{+}||,\sqrt{1/N})}\bigl[\mathbf{u}_{m}\bigr]_{+},\quad \forall m=1,2,\ldots,M.
	\end{equation}
\end{proposition}
\begin{IEEEproof}
	See Appendix \ref{sec_proj}.
\end{IEEEproof}

\subsection{Proposed Algorithm with Line Search}
In Algorithm~\ref{alg:APG} is guaranteed to converge for any fixed  step sizes smaller than $L_{f}$. However, it is possible that  $L_{f}$ given in \eqref{eq_lf1} is  significantly larger than the best Lipschitz constant of the gradient of $f_{\xi}(\pmb{\theta})$ which is practically difficult to find. In order to find a larger step size, and thus faster convergence, we can carry out a line search to tune the step size in \eqref{eq:updatez} and \eqref{eq:updatev}. In this paper, we can perform a line search as described in Algorithm
\ref{alg:Backtracking-linear-search}, inspired from  \cite{Li:2015:APG}, which works as follows. In each iteration, the backtracking line search  starts with a large step size, and then decrease
it until a better feasible solution is found. 
 As we can see from Algorithm \ref{alg:Backtracking-linear-search}, the algorithm will always terminate with a better point, in the sense of maximizing the  objective function $f_{\xi}(\pmb{\theta})$. Note that the backtracking line search in Algorithm \ref{alg:Backtracking-linear-search}
follows the Barzilai-Borwein (BB) rule \cite{barzilai1988two}. As $\nabla f_{\xi}(\pmb{\theta})$
is Lipschitz continuous with a Lipschitz constant $L_{f}$ given in \eqref{eq_lf1}, the
line search procedure is guaranteed to terminate after finite steps.

\begin{algorithm}[!t]
\hrulefill{}
\caption{Backtracking line search for finding a step size for \eqref{eq:updatez} and \eqref{eq:updatev}\label{alg:Backtracking-linear-search}}

\SetAlgoNoLine
\SetNlSty{textnormal}{}{}

\KwIn{ $\nu<1$, $\delta>0$} 

		$\mathbf{s}^{(n)}=\mathbf{z}^{(n)}-\mathbf{y}^{(n-1)}$; $\mathbf{r}^{(n)}=\nabla f_{\xi}(\mathbf{z}^{(n)})-\nabla f_{\xi}(\mathbf{y}^{(n-1)})$

Set $\alpha_{y}=\frac{(\mathbf{s}^{(n)})^{T}\mathbf{s}^{(n)}}{(\mathbf{s}^{(n)})^{T}\mathbf{r}^{(n)}}$
or $\alpha_{y}=\frac{(\mathbf{s}^{(n)})^{T}\mathbf{r}^{(n)}}{(\mathbf{r}^{(n)})^{T}\mathbf{r}^{(n)}}$,

$\mathbf{s}^{(n)}=\mathbf{v}^{(n)}-\pmb{\theta}^{(n-1)}$;
$\mathbf{r}^{(n)}=\nabla f_{\xi}(\mathbf{v}^{(n)})-\nabla f_{\xi}(\pmb{\theta}^{(n-1)})$,

Set $\alpha_{\theta}=\frac{(\mathbf{s}^{(n)})^{T}\mathbf{s}^{(n)}}{(\mathbf{s}^{(n)})^{T}\mathbf{r}^{(n)}}$
or $\alpha_{\theta}=\frac{(\mathbf{s}^{(n)})^{T}\mathbf{r}^{(n)}}{(\mathbf{r}^{(n)})^{T}\mathbf{r}^{(n)}}$,


\Repeat(\tcc*[f]{step size for \eqref{eq:updatez}}){$f_{\xi}(\mathbf{z}^{(n+1)})\geq f_{\xi}(\mathbf{y}^{(n)})+\delta||\mathbf{z}^{(n+1)}-\mathbf{y}^{(n)}||^{2}$}{
	
	$\mathbf{z}^{(n+1)}=P_{\mathcal{C}}(\mathbf{y}^{(n)}+\alpha_{y}\nabla f_{\xi}(\mathbf{y}^{(n)}))$,
	
	$\alpha_{y}=\alpha_{y}\nu$,
	
}
\Repeat(\tcc*[f]{step size for \eqref{eq:updatev}}){$f_{\xi}(\mathbf{v}^{(n+1)})\geq f_{\xi}(\pmb{\theta}^{(n)})+\delta||\mathbf{v}^{(n+1)}-\pmb{\theta}^{(n)}||^{2}$}{
	
	$\mathbf{v}^{(n+1)}=P_{\mathcal{C}}(\pmb{\theta}^{(n)}+\alpha_{\theta}\nabla f_{\xi}(\pmb{\theta}^{(n)}))$,\label{theta:update}
	
	$\alpha_{\theta}=\alpha_{\theta}\nu$,
	
}
\hrulefill{}\\ ~
\end{algorithm}

\subsection{Convergence Analysis of Proposed Method}
The convergence of Algorithm \ref{alg:APG} is guaranteed that of the APG method and the penalty method. Specifically,
 for a given $\xi$, similar to  Theorem 1 of \cite{Li:2015:APG}, we can show that the 
 objective sequence $\{f_{\xi }(\boldsymbol{\theta}^{(n)})\} $ is monotonically increasing.
Also, the sequence $\{\boldsymbol{\theta}^{(n)}\}$ is bounded and thus has accumulation points. Each accumulation point is also a stationary solution to \eqref{eq:opt_theta}.
Furthermore, following the arguments in \cite[Chap. 10]{AkardiOptII} we can show that the iterate sequence $\{\boldsymbol{\theta}_{\xi_{m}}\}$ converges (in the subsequence sense) to a feasible point of $(\mathcal{P}_2)$ when $\xi_{m} \to \infty$. Thus the obtained solution of Algorithm \ref{alg:APG} is also a stationary point of $(\mathcal{P}_2).$
%
The proof of these claims is given in Appendix \ref{Appe1}. We note however that since Algorithm \ref{alg:APG} will terminate for some finite $\xi_{m}$ when a pre-determined error tolerance is met, it can only produce an approximate stationary solution of $(\mathcal{P}_2)$.
\subsection{Computational Complexity Analysis}
It is obvious that the complexity of Algorithm \ref{alg:APG} in each iteration is dominated by that of \eqref{eq:updatez}, and \eqref{eq:updatev} and the computation of the objective. Here we use the big-$\mathcal{O}$ notation to analyse the complexity of  \eqref{eq:updatez} and \eqref{eq:updatev}.  From Proposition \ref{lem:grad}, it is easy to see that the complexity to calculate the gradient of $\nabla f_{\xi}(\pmb{\theta})$ is $\mathcal{O}(KM^2)$. The projection $P_{\mathcal{C}}(\mathbf{u})$ requires the complexity of $\mathcal{O}(KM)$ which is obvious from Proposition \ref{lem:proj}. Similarly, the complexity of computing $f_{\xi}(\pmb{\theta})$ is $\mathcal{O}(KM^2)$. As a result, the overall complexity of Algorithm \ref{alg:APG} is $\mathcal{O}(I_{P}I_{APG}KM^2)$ where 
 $I_{P}$ and $I_{APG}$ are the number of iterations of the outer loop (i.e. the penalty method)  and the inner loop (i.e. the APG method) in Algorithm \ref{alg:APG}, respectively.  Note that the line search procedure contributes negligible complexity since the gradient can be reused and the projection requires much less complexity. 
In \cite{Ngo2018a}, a SCA method based on solving a sequence of SOCPs was presented. We remark that the complexity of solving an SOCP in each SCA iteration is $\mathcal{O}(\sqrt{K+M}M^3 K^4)$ \cite{ben2001lectures}. Thus the complexity of the SCA method  in  \cite{Ngo2018a} is  $\mathcal{O}(I_{SCA}\sqrt{K+M}M^3 K^4)$ where $I_{SCA}$ is the number of SCA iterations. It is apparent that the  computational complexity of our proposed method is much lower than the SCA method in \cite{Ngo2018a}.This point is numerically demonstrated in the next section.

\section{Numerical Results}\label{sec:numresults}

In this section, numerical results will be provided to evaluate as well as show the benefits of our proposed algorithm. 
\subsection{System Setup}
We consider cell-free massive MIMO systems, where
locations of $M$ APs and $K$ users randomly uniformly generated within an area of $1\times 1$
km$^2$. The wrapped around technique is used. The large-scale fading coefficient is modeled as:
\begin{equation} 
\beta _{mk} = \text {PL}_{mk}\cdot z_{mk}, 
\end{equation}
where $z_{mk}$ is the log-normal shadowing with the standard derivation $\sigma_{sh} = 8$ dB, and $\text {PL}_{mk}$ is the three-slope-based path loss, which is modeled (in dB) as
\begin{align} 
\text {PL}_{mk} = \begin{cases} -L - 35\log _{10} (d_{mk}),~\text {if}~d_{mk}>d_{1}\\[4pt] -L - 15\log _{10} (d_{1}) - 20\log _{10} (d_{mk}), ~\text {if}~d_{0}< d_{mk}\leq d_{1}\\[4pt] -L - 15\log _{10} (d_{1}) - 20\log _{10} (d_{0}), ~ \text {if}~d_{mk} \leq d_{0},\\[4pt] \end{cases}
\end{align}
where we choose $d_{0} = 10$~m, $d_{1} = 50$~m, and $L = 140.7$~dB. The power consumption is summarized as follows: power amplifier coefficient $a_{m} = 0.4,~\forall m$; internal power consumption per antenna $P_{\text{tc},m} = 0.2~\forall m$; fixed power consumption per each backhaul. In addition, we choose ${\So}_{k} = 1$ bit/s/Hz, $B=20$~MHz, $\Pd = 1$~W, $\Pp=0.2$~W, and noise figure is 9~dB. In simulation, we implement  Algorithm \ref{alg:APG} with the line search described in Algorithm \ref{alg:Backtracking-linear-search}  where $\rho=0.5$.
\subsection{ Convergence of  Proposed Algorithm}
In the first experiment, we show the performance of the proposed APG method in comparison with the sequential SOCPs-based method in \cite{Ngo2018a}, with $D=1$~km, $\tau_c = 200$, and $\tau_p = K$. As can be seen in Fig. \ref{fig:convergence}, the proposed method  achieves
the same EE performance as the SOCPs-based method. In terms of the number
of iterations required to output a solution, although our proposed APG method requires more iterations to converge, compared to the SOCPs-based method. However, as we mentioned previously, the proposed method
requires very cheap iteration cost, and thus is far more efficient
in terms of the actual run time. This point is clearly illustrated in
Fig. \ref{fig:runtime}, where we plot the run time of the proposed
algorithm and the SOCPs-based method as a function of $M$. The simulations are built using  MATLAB and the results
are obtained on a Dell laptop with Intel Core\texttrademark{} i7-9750H and
RAM of 16 GB. The stopping criterion is $\varsigma = 10^{-3}$. Compared to the sequential SOCPs-based method, our proposed scheme reduces the run time significantly, i.e, about $62$ times and $53$ times when $M=100$ and $M=400$, respectively. 

Next, we need to verify that the proposed APG algorithm will not violate any PFs, or equivalently, the total loss, $\sum_k^K\Psi_{k}(\pmb{\theta})$, will converge to $0$, regardless the starting point. The numerical results are shown in Fig. \ref{fig:convergence_vs_PFs400} using two different scenarios with the number of APs, $M = 100$ and $M = 400$, respectively.  In this figure, blue curves and orange curves represent total PFs (total loss) and total EE, respectively. 
In the both scenarios, total PFs starts in infeasible domain, and gradually converges to $0$  when the algorithm terminate, as the result of increasing the penalty parameter $\xi$.
\begin{figure}[h!]
	\centering
	\includegraphics[width=0.65\columnwidth]{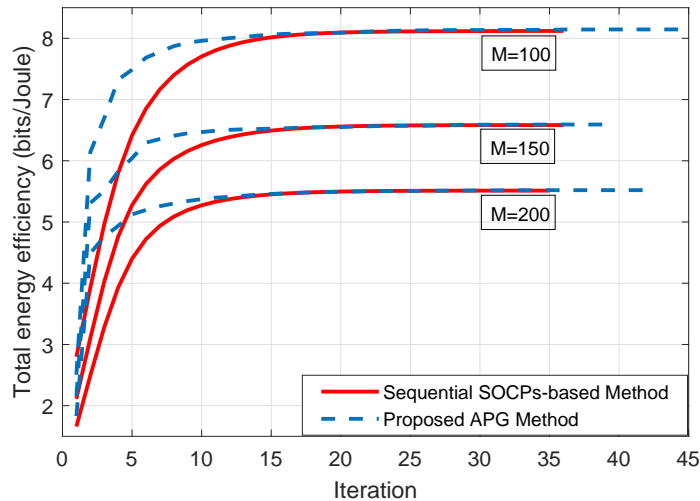}
	
	\caption{Convergence rate of the proposed APG algorithm with $K$ = 40, $N$ = 1.}
	\label{fig:convergence}
\end{figure}
\begin{figure}[h!]
	\centering
	\includegraphics[width=0.65\columnwidth]{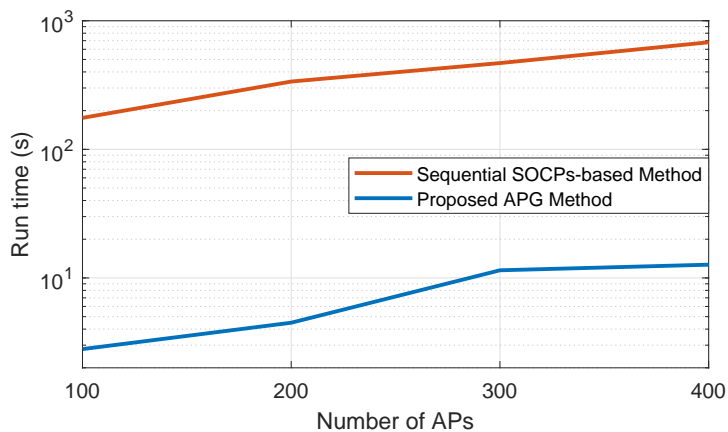}\caption{Run time versus the number of APs $M$. The
		number of users is $K=40$.}
	\label{fig:runtime}
\end{figure}
\begin{figure}[h!]
	\centering
	\includegraphics[width=0.65\columnwidth]{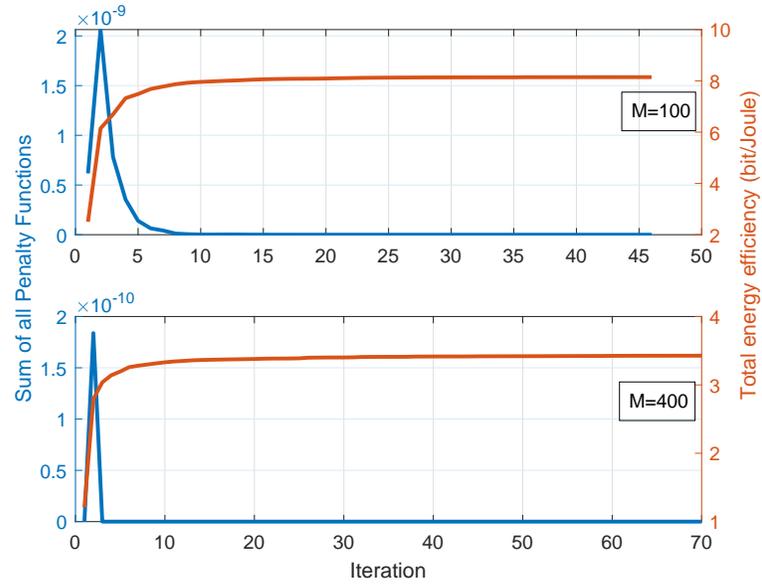}
	
	\caption{Convergence rate vs sum of all penalty functions with $K$ = 40, $N$ = 1.}
	\label{fig:convergence_vs_PFs400}
\end{figure}
%
\subsection{Multi-antenna APs}
\begin{figure}[h!]
	\centering
	\subfigure[]{\includegraphics[width=0.49\textwidth]{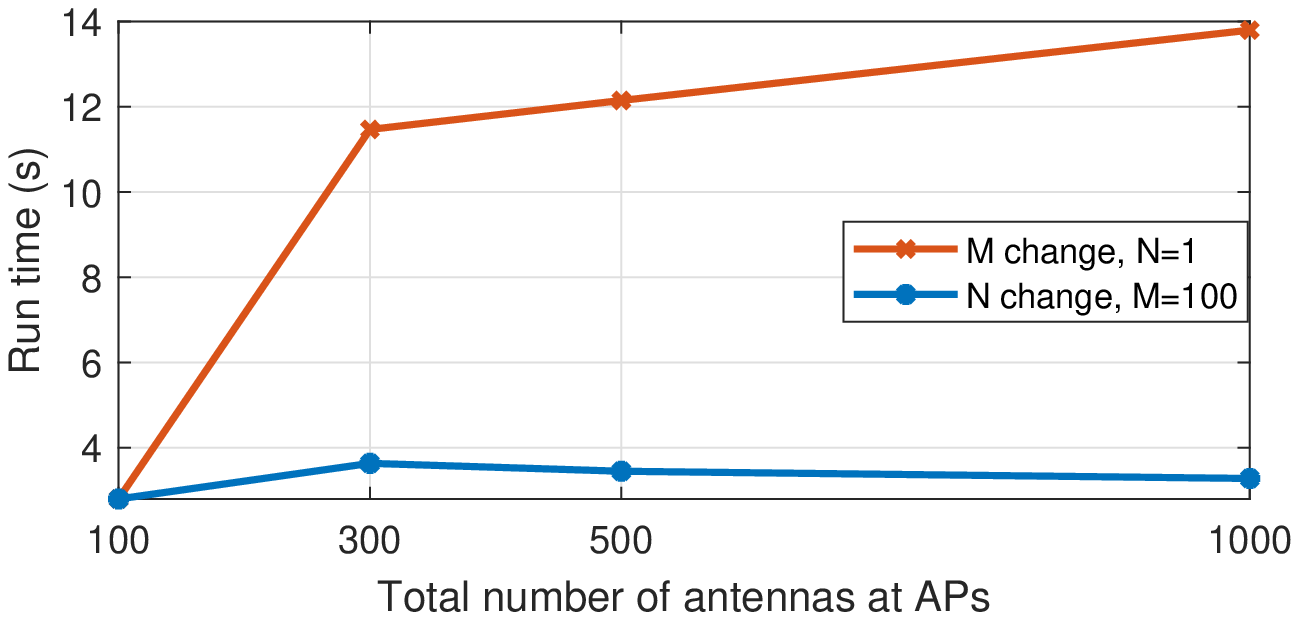}} 
	\subfigure[]{\includegraphics[width=0.49\textwidth]{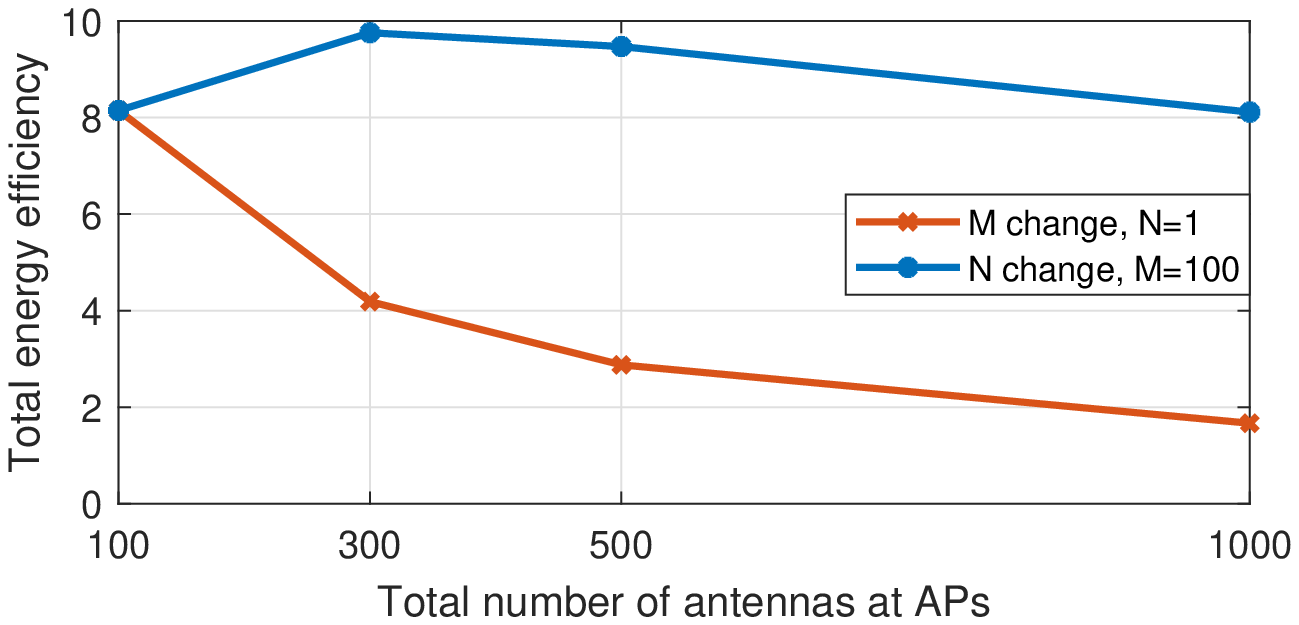}} 
	\caption{Run time and total energy efficiency vs total number of antennas at APs with $K=40$.}
	\label{fig:runtime_M_vs_N}
\end{figure}

Fig. \ref{fig:runtime_M_vs_N} examines the effect of multiple antennas at the APs. We consider two scenarios: ($N=1$ and $M$ changes) and ($N$ changes, $M=100$). For a fair comparison, both scenarios have the same total number of  antennas of all APs, i.e. $MN$ is fixed. The numerical results show that, with fixed number of APs $M =100$, run time just changes a small amount, or even faster when changing the number of antennas per AP from 1 to 10. However, run time increase proportionally with the number of APs using single antenna. This is indeed an expected result since the complexity of our proposed algorithm only depends on the number of APs $M$.  Regarding to the total EE, Fig. \ref{fig:runtime_M_vs_N} (b) shows that multi-antenna APs always outperform single-antenna APs, on the condition that they have the same total number of antennas at APs. The main reason is that more energy is consumed when single-antena's AP is used. Therefore, based on our numerical results, instead of increasing the number of APs, we should increase the number of anntennas per APs to take the advantages of both run time and total EE of the proposed algorithm. 
\subsection{System setup based on number of users}
\begin{figure}[h!]
	\centering
	\subfigure[]{\includegraphics[width=0.49\textwidth]{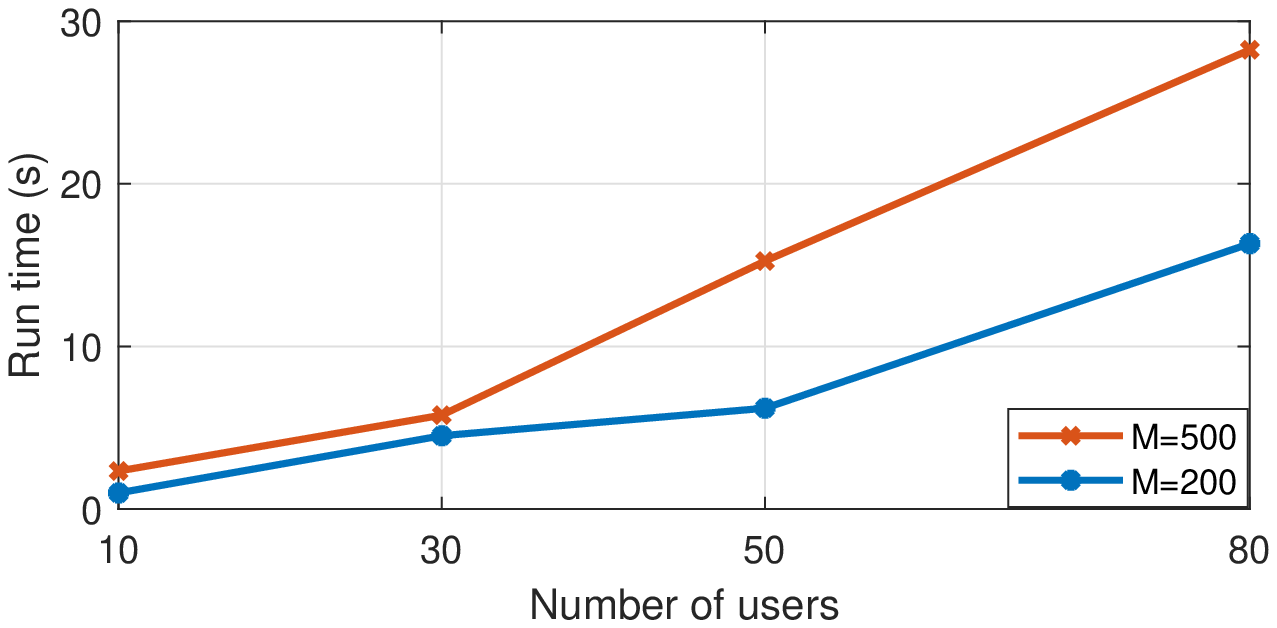}} 
	\subfigure[]{\includegraphics[width=0.49\textwidth]{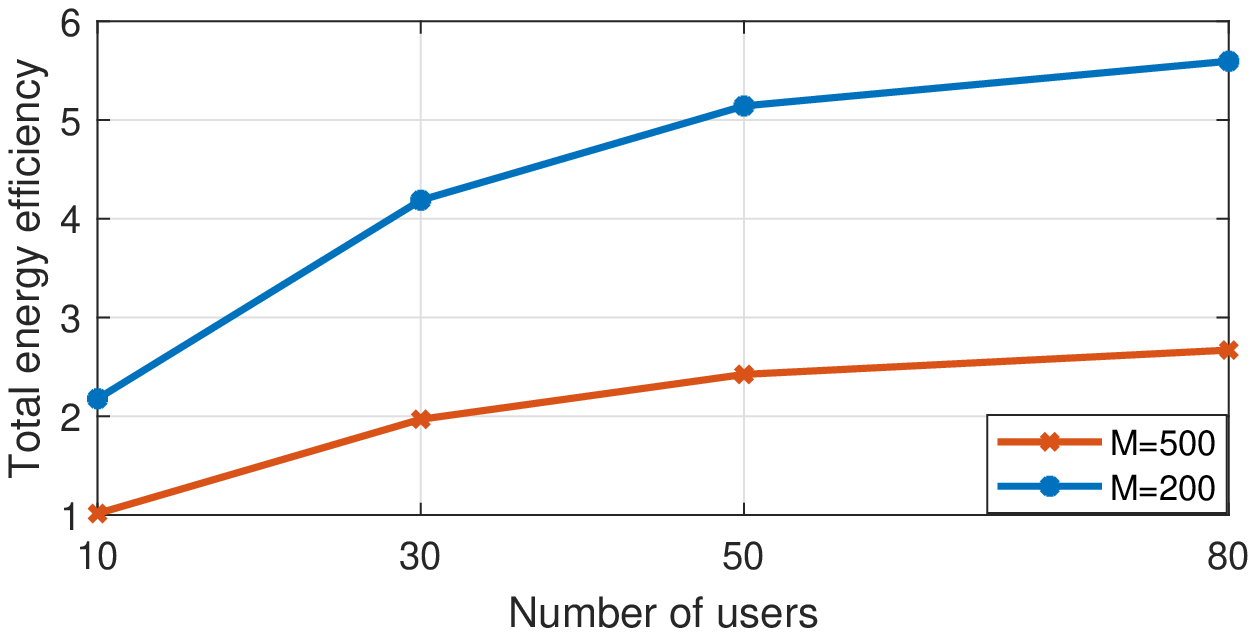}} 
	\caption{Run time and total energy efficiency vs number of users with $N=10$.}
	\label{fig:run_time_Kchange}
\end{figure}

%
Fig. \ref{fig:run_time_Kchange} compares two setups of the system with different number of APs using multi-antennas. It clear that both system's performance, i.e. total EE, and run time heavily depend on number of APs in the system regardless number of users. In this case, we can see that system with $M=200$ always outperforms the one with $M=500$ in term of run time and total EE, when the number of users changes from $10$ to $80$. Therefore, based on number of users in the system, we can setup the system by just activating the suitable number of APs to achieve higher performance and faster running time.
\subsection{System Scale V.s. Total energy efficiency}
\begin{figure}[h!]
	\centering
	\includegraphics[width=0.6\columnwidth]{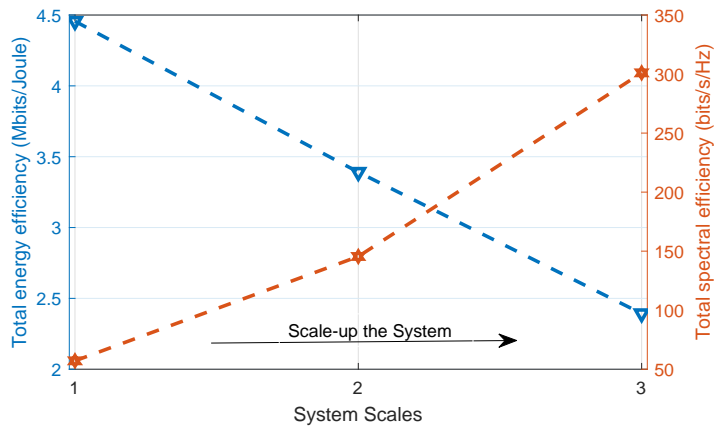}\caption{Total energy efficiency and total spectral efficiency vs system scale with $N=1$.}
	\label{fig:EE_compare_systemScale}
\end{figure}
Fig. \ref{fig:EE_compare_systemScale} considers three different system scale setups, where the first, the second, and the third system scales are corresponding to different number of APs and users $M \times K$, which are $200 \times 40$, $500 \times 100$, and $1000 \times 200$, respectively. We choose $\tau_c = 500$, and $\tau_p = 200$ for all scenarios. The result shows the trade-off between the spectral efficiency and energy efficiency. While the total SE proportionally increases when extending the system scale (because the number of users increases), the total EE shows the opposite trend (because the energy consumption increases). Depending on practical applications and requirements, suitable numbers of APs and users can be designed.
%

\section{Conclusion \label{sec:conclusion}}

We applied a APG method to deal with a large-scale EE optimization problem in cell-free massive MIMO systems, taking into account of power contraints at the APs and QoS contraints at each user. To compare with  
the sequential SOCPs-based approaches, our proposed method achieves the same performance, while its run time is much faster, i.e., one to two orders of magnitude reduction. Therefore, this method has a great potential to deal with large-scale optimization problems in cell-free massive MIMO, and hence, can be applied to practical scenarios. In addition, our optimization problem is done over large-scale fading time scale. It means the optimal power control coefficients can be updated only once for each large-scale fading realization. Since the large-scale fading coefficient changes very slowly with time, the system can fully operate in a bursty communication scenario with some random activation of the users.

\appendix{}
\subsection{Useful properties and Lemmas}\label{sec_upl}
In this section, we provide some useful properties and lemmas related to Lipschitz continuity that shall be used to analyse the Lipschitz continuity of the gradient of the objective in \eqref{eq:opt_theta}. Some of the following properties and lemmas are in fact extension of the results in \cite[Section 1.5]{Weaver:2018:Lipschitz} for scalar-valued functions. We also provide proofs to these for the sake of completeness.

\subsubsection{Linear Combinations}\label{sec_lc}

Let $f_{1}:\mathbb{R}^n\to\mathbb{R}^n$ and $f_{2}:\mathbb{R}^{n}\to \mathbb{R}^{n}$ be Lipschitz continuous with a Lipschitz constants
$L_{1}$ and $L_{2}$, respectively. Then $f_{1}\pm f_{2}$ is Lipschitz
continuous with a Lipschitz constant $L_{1}+L_{2}$.

\begin{IEEEproof}
	\begin{equation*}
	\begin{aligned}
	||(f_1+f_2)(\mathbf{x})-(f_1+f_2)(\mathbf{y})|| \mathop \leq \limits ^{(a)}||f_1(\mathbf{x})-f_1(\mathbf{y})||+||f_2(\mathbf{x})-f_2(\mathbf{y})|| \leq \left(L_1+L_2\right) ||\mathbf{x}-\mathbf{y}||,
	\end{aligned}
	\end{equation*}
	where $(a)$ is the triangle inequality.
\end{IEEEproof}
%

\subsubsection{Product of Functions}\label{prod_lipz}

Let $f_{1}:\mathbb{R}^n\to\mathbb{R}$ and $f_{2}:\mathbb{R}^{n}\to \mathbb{R}^{m}$ be bounded functions over a bounded set
$\mathcal{S}$, such that
$|f_{1}(\mathbf{x})|\leq M_1$ and $||f_{2}(\mathbf{x})||\leq M_2$ for
all $\mathbf{x}\in\mathcal{S}$. Further suppose that $f_{1}$ and
$f_{2}$ are Lipschitz continuous with a Lipschitz constants $L_{1}$
and $L_{2}$, respectively. Then the product of $f_{1}$ and $f_{2}$
is Lipschitz continuous with a Lipschitz constant $(M_1L_1 + M_2L_2)$. 
\begin{IEEEproof}
	\begin{equation*}
	\begin{aligned}
	||f_1f_2(\mathbf{x})-f_1f_2(\mathbf{y})||
	&\leq||f_1(\mathbf{x})\left(f_2(\mathbf{x})-f_2(\mathbf{y})\right) + \left(f_1(\mathbf{x})-f_1(\mathbf{y})\right)f_2(\mathbf{y})||\\
	&\leq||f_1(\mathbf{x})\left(f_2(\mathbf{x})-f_2(\mathbf{y})\right)|| + ||\left(f_1(\mathbf{x})-f_1(\mathbf{y})\right)f_2(\mathbf{y})||\\
	&\leq M_1||f_2(\mathbf{x})-f_2(\mathbf{y})|| + M_2||f_1(\mathbf{x})-f_1(\mathbf{y})|| \leq (M_1L_1 + M_2L_2) ||\mathbf{x}-\mathbf{y}||.
	\end{aligned}
	\end{equation*}
\end{IEEEproof}
\subsubsection{The Composition of Functions}\label{subsec_compfunc}

Let $f_{1}:\mathbb{R}\to\mathbb{R}$ and $f_{2}:\mathbb{R}^{n}\to \mathbb{R}$
be Lipschitz continuous with Lipschitz constants $L_{1}$ and $L_{2}$,
respectively. Then the composite function $f_{1}\circ f_{2}$ is Lipschitz
continuous with a Lipschitz constant $L_{1}L_{2}$. 
\begin{IEEEproof}
	\begin{equation*}
	\begin{aligned}
	||f_{1}(f_{2}(\mathbf{x}))-f_{1}(f_{2}(\mathbf{y}))|| \leq L_{1}||f_{2}(\mathbf{x})-f_{2}(\mathbf{y})||\leq L_{1}L_{2}||\mathbf{x}-\mathbf{y}||.
	\end{aligned}
	\end{equation*}
\end{IEEEproof}
\subsubsection{The Quotient of Functions}\label{quoti_subsec}

Let $f_{1}:\mathbb{R}^n\to\mathbb{R}^n$ and $f_{2}:\mathbb{R}^{n}\to \mathbb{R}$ be Lipschitz continuous with Lipschitz constants
$L_{1}$ and $L_{2}$, respectively on a bounded set $\mathcal{S}$, such that
$||f_{1}(\mathbf{x})||\leq M$ for
all $\mathbf{x}\in\mathcal{S}$,
and further assume that there is a constant $c>0$ such that $|f_{2}(\mathbf{x})|\geq c$
for all $\mathbf{x}\in\mathcal{S}$. Then $f_{1}/f_{2}$ is Lipschitz
continuous with a Lipschitz constant $(ML_1 + L_2/c^3)$. 

\begin{IEEEproof}
	First, since $f_{2}$ is Lipschitz continuous with a Lipschitz constants $L_{2}$, and $|f_{2}(\mathbf{x})|\geq c$, we have
	\begin{equation*}
	\begin{aligned}
	|1 / f_2(x)-1 / f_2(y)| =\frac{|f_2(x)-f_2(y)|}{|f_2(x) f_2(y)|} \leq \frac{L_{2}}{c^{2}}||\mathbf{x}-\mathbf{y}||.
	\end{aligned}
	\end{equation*}
	Thus $1 / f_2$ is Lipschitz continuous with a Lipschitz constant $\frac{L_{2}}{c^{2}}$. Next, applying the product property in \ref{prod_lipz} for $f_1$ and $ 1 / f_2$,  we have 
	\begin{equation*}
	\begin{aligned}
	||(f_1/f_2)(\mathbf{x})-(f_1/f_2)(\mathbf{y})||\le (ML_1 + L_2/c^3) ||\mathbf{x}-\mathbf{y}||.
	\end{aligned}
	\end{equation*}
\end{IEEEproof}

\begin{lemma}\label{lem1}
	Let $f:\mathbb{R}^n\to\mathbb{R}$, and assume that $||\mathbf{x} + \mathbf{y}||  \le c,$ $\forall \mathbf{x},\mathbf{y}  \in\mathcal{S}$. Then  $f(\mathbf{x}) = (\mathbf{a}^T\mathbf{x})^2$ is Lipschitz continuous with a Lipschitz constant $c||\mathbf{a}||^2$.
\end{lemma}
\begin{IEEEproof}
	\begin{align}\label{gxminusgy}
	|f(\mathbf{x}) - f(\mathbf{y})|  &= | (\mathbf{a}^T\mathbf{x})^2 - (\mathbf{a}^T\mathbf{y})^2 | = | \left(\mathbf{a}^T\mathbf{x} - \mathbf{a}^T\mathbf{y} \right)\left(\mathbf{a}^T\mathbf{x} + \mathbf{a}^T\mathbf{y} \right)|\notag\\
	&\mathop \leq \limits ^{(b1)} \left|\mathbf{a}^T\mathbf{x} + \mathbf{a}^T\mathbf{y} \right|||\mathbf{a}|||| \mathbf{x} - \mathbf{y}||\mathop \leq \limits ^{(b2)} c||\mathbf{a}||^2|| \mathbf{x} - \mathbf{y}||,
	\end{align}
	where $(b1)$ and $(b2)$ base on  the Cauchy-Schwarz inequality.
\end{IEEEproof}

\begin{lemma}\label{lem2}
	Let $f:\mathbb{R}^n\to\mathbb{R}$, and assume that $\left(||\mathbf{a}\odot\mathbf{x}|| + ||\mathbf{a}\odot\mathbf{y}||\right)   \le d,$ $\forall \mathbf{x},\mathbf{y} \in \mathbb{R}^n$, then a Lipschitz constant of  $f(\mathbf{x}) = ||\mathbf{a}\odot\mathbf{x}||^2$ is Lipschitz continuous with a Lipschitz constant $d||\mathbf{a}||$.
\end{lemma}
\begin{IEEEproof}
	\begin{align}
	||f(\mathbf{x}) - f(\mathbf{y})|| 
	& = \left| ||\mathbf{a}\odot\mathbf{x}||^2 - ||\mathbf{a}\odot\mathbf{y}||^2 \right| = \left| \left(||\mathbf{a}\odot\mathbf{x}|| - ||\mathbf{a}\odot\mathbf{y}||\right) \left( ||\mathbf{a}\odot\mathbf{x}|| + ||\mathbf{a}\odot\mathbf{y}||\right)\right|\notag\\
	& \mathop \leq \limits ^{(c1)} ||\mathbf{a}\odot\mathbf{x}-\mathbf{a}\odot\mathbf{y}|| \left( ||\mathbf{a}\odot\mathbf{x}|| + ||\mathbf{a}\odot\mathbf{y}||\right) \mathop \leq \limits ^{(c2)} d||\mathbf{a}|| ||\mathbf{x} - \mathbf{y}||,
	\end{align}
	where $(c1)$ and  $(c2)$ are due to the triangle inequality and the Cauchy-Schwarz inequality, respectively.
\end{IEEEproof}
\subsection{Proof of Proposition \ref{pp1}}\label{sec_lips_proof}
First, note that $f_{\xi}(\pmb{\theta})$ is proper as $\text{dom} f_{\xi} \ne 0$. Next, we need to prove that $f_{\xi}(\pmb{\theta})$ is Lipschitz continuous gradient. $\nabla f_{\xi}(\pmb{\theta})$ can be calculated as
\begin{align}
\nabla f_{\xi}(\pmb{\theta}) =B\frac{v(\pmb{\theta})\nabla u(\pmb{\theta})-u(\pmb{\theta})\nabla v(\pmb{\theta})}{v(\pmb{\theta})^{2}} - 2\xi\sum_{k=1}^{K}\left[\max\left(0, { g}_{k}(\pmb{\theta}) \right )\right] \nabla g_{k}(\pmb{\theta}).\label{eq:grad:obj}
\end{align}
As $\nabla f_{\xi}(\pmb{\theta})$ is computed from $u_k(\pmb{\theta})$, $v(\pmb{\theta})$, ${ g}_{k}(\pmb{\theta})$, $\nabla u_k(\pmb{\theta})$, $\nabla v(\pmb{\theta})$, and $\nabla g_{k}(\pmb{\theta})$ we now need to find the Lipschitz constants of these terms, and then apply properties in Appendix \ref{sec_upl} to conclude the Lipschitz constant of $\nabla f_{\xi}(\pmb{\theta})$. To this end the following results are in order

\subsubsection{ $\nabla u_k(\pmb{\theta})$ is Lipschitz continuous, and $||\nabla u_k(\pmb{\theta})||$ is bounded from above}\label{sec_nuk}

Recall that 
\begin{align}\label{eq_guk}
\nabla u_k(\pmb{\theta}) = \frac{\nabla n_k(\pmb{\theta})}{\text{ln}2\left(n_k(\pmb{\theta}) + d_k(\pmb{\theta})\right)} - \frac{n_k(\pmb{\theta}) \nabla d_k(\pmb{\theta})}{\text{ln}2\left(n_k(\pmb{\theta}) + d_k(\pmb{\theta})\right)d_k(\pmb{\theta})},
\end{align}
where
\begin{equation}
n_k(\pmb{\theta}) \triangleq \Pd N^{2}\bigl(\tilde{\pmb{\gamma}}_{kk}^{T}\mathbf{A}_{k}\pmb{\theta}\bigr)^{2},
\end{equation}
\begin{equation}
d_k(\pmb{\theta}) \triangleq \Pd N^{2}\sum\limits _{k'\neq k}^{K}\bigl(\tilde{\pmb{\gamma}}_{k'k}^{T}\mathbf{A}_{k'}\pmb{\theta}\bigr)^{2}+\Pd N\sum\limits _{k'=1}^{K}||\tilde{\pmb{\kappa}}_{k}\odot(\mathbf{A}_{k'}\pmb{\theta})||^{2}+1,
\end{equation}

\begin{align}
\nabla n_k(\pmb{\theta}) = 2\Pd N^{2}\mathbf{A}_{k}^{T}\tilde{\pmb{\gamma}}_{kk}\tilde{\pmb{\gamma}}_{kk}^{T}\mathbf{A}_{k}\pmb{\theta},
\end{align}
and 
\begin{align}
\nabla d_k(\pmb{\theta}) = 2\Pd N^{2}\sum_{k'\ne k}^{K}\mathbf{A}_{k'}^{T}\tilde{\pmb{\gamma}}_{k'k}\tilde{\pmb{\gamma}}_{k'k}^{T}\mathbf{A}_{k'}\pmb{\theta} + 2\Pd N\sum_{k'=1}^{K}\mathbf{A}_{k'}^{T}\mathbf{B}_{k}\mathbf{A}_{k'}\pmb{\theta}.
\end{align}
In order to show that $\nabla u_k(\theta)$ is Lipschitz continuous, we will prove that the first term and the second term of the right hand side in (\ref{eq_guk}) are Lipschitz continuous, respectively. First, note that by applying Lemma \ref{lem1} and Lemma \ref{lem2}, it is easy to see that $n_k(\pmb{\theta}),~ d_k(\pmb{\theta})$ and $n_k(\pmb{\theta}) + d_k(\pmb{\theta})$ are Lipschitz continuous and their Lipschitz constants are 
\begin{align}\label{eq_lc1}
L_{n_k} =2\Pd N^{3/2}\bigl||\tilde{\pmb{\gamma}}_{kk}^{T}\mathbf{A}_{k}||^2,
\end{align}
\begin{align}\label{eq_lc2}
L_{d_k} =2\Pd N^{3/2}\sum_{k'\ne k}^K\bigl||\tilde{\pmb{\gamma}}_{k'k}^{T}\mathbf{A}_{k'}||^2 + 2\sqrt{K/N}\tilde{\pmb{\kappa}}_{k}^T \mathbf{1}_M||\tilde{\pmb{\kappa}}_{k}||,
\end{align}
and
\begin{align}\label{eq_lc3}
L_{n_k,d_k} =2\Pd N^{3/2}\sum_{k=1}^K\bigl||\tilde{\pmb{\gamma}}_{kk}^{T}\mathbf{A}_{k}||^2 + 2\sqrt{K/N}\tilde{\pmb{\kappa}}_{k}^T \mathbf{1}_M||\tilde{\pmb{\kappa}}_{k}||,
\end{align}
respectively. Next, we have $
\nabla n_k(\pmb{\theta}) = 2\Pd N^{2}\mathbf{A}_{k}^{T}\tilde{\pmb{\gamma}}_{kk}\tilde{\pmb{\gamma}}_{kk}^{T}\mathbf{A}_{k}\pmb{\theta} \triangleq \mathbf{C}_{n_k}\pmb{\theta}$
is Lipschitz continuous, as
\begin{align}\label{eq_lgu}
||\nabla n_k(\mathbf{x}) - \nabla n_k(\mathbf{y})|| = ||\mathbf{C}_{n_k}(\mathbf{x} - \mathbf{y})|| \le \lambda_{\text{max}}(\mathbf{C}_{n_k})||\mathbf{x} - \mathbf{y}||,
\end{align}
where $\lambda_{\max }(\mathbf{C}_{n_k})$
 is the largest eigenvalue of $\mathbf{C}_{n_k}$. Moreover,
\begin{align}\label{eq_uu}
||\nabla n_k(\pmb{\theta})|| = ||\mathbf{C}_{n_k}\pmb{\theta}||\le \frac{\lambda_{\text{max}}(\mathbf{C}_{n_k})\sqrt{M}}{\sqrt{N}}, 
\end{align}
and
\begin{align}\label{eq_lbu}
\text{ln}2\left(n_k(\pmb{\theta}) + d_k(\pmb{\theta})\right) \ge \ln 2.
\end{align}
Then, from (\ref{eq_lc3}), (\ref{eq_lgu}), (\ref{eq_uu}), and (\ref{eq_lbu}), by applying the quotient of functions in Section \ref{quoti_subsec}, we have that 
$\frac{\nabla n_k(\pmb{\theta})}{\text{ln}2\left(n_k(\pmb{\theta}) + d_k(\pmb{\theta})\right)}$ is Lipschitz continuous with a Lipschitz constant as follow
\begin{align}\label{eq_lc4}
L_{t1} =\frac{(\lambda_{\text{max}}(\mathbf{C}_{n}))^2\sqrt{M}}{\sqrt{N}} + \frac{L_{n,d}}{(\ln 2)^2}.
\end{align}
By applying the similar method to $\nabla n_k(\pmb{\theta})$, we can prove that $\nabla d_k(\pmb{\theta})$ is Lipschitz continuous with a Lipschitz constant written as
\begin{align}\label{eq_lc5}
L_{\nabla d_k} = \lambda_{\text{max}}(\mathbf{C}_{d_k}),
\end{align}
where $\mathbf{C}_{d_k} \triangleq 2\Pd N^{2}\sum_{k'\ne k}^{K}\mathbf{A}_{k'}^{T}\tilde{\pmb{\gamma}}_{k'k}\tilde{\pmb{\gamma}}_{k'k}^{T}\mathbf{A}_{k'} + 2\Pd N\sum_{k'=1}^{K}\mathbf{A}_{k'}^{T}\mathbf{B}_{k}\mathbf{A}_{k'}$. Moreover, we have
\begin{align}\label{eq_uu1}
||\nabla d_k(\pmb{\theta})|| \le \frac{\lambda_{\text{max}}(\mathbf{C}_{d_k})\sqrt{M}}{\sqrt{N}}, 
\end{align}
and
\begin{align}\label{eq_uu2}
|n_k(\pmb{\theta})| \mathop < \limits ^{(d1)} n_{\text{max}} < \infty,
\end{align}
where $(d1)$ is based on the fact that $|n_k(\pmb{\theta})|$ is bounded by the power constraints. Then, from (\ref{eq_lc1}), (\ref{eq_lc5}), (\ref{eq_uu1}), and (\ref{eq_uu2}), by applying the product of functions in Section \ref{prod_lipz}, it is easy to see  that $n_k(\pmb{\theta}) \nabla d_k(\pmb{\theta})$ is Lipschitz continuous with a Lipschitz constant given by
\begin{align}\label{eq_uu3}
L_{n_k,\nabla d_k} =  n_{\text{max}}L_{n_k} + \frac{(\lambda_{\text{max}}(\mathbf{C}_{d_k}))^2\sqrt{M}}{\sqrt{N}},
\end{align}
and
\begin{align}\label{eq_uu4}
||n_k(\pmb{\theta}) \nabla d_k(\pmb{\theta})|| \le  \frac{n_{\text{max}}\lambda_{\text{max}}(\mathbf{C}_{d_k})\sqrt{M}}{\sqrt{N}}.
\end{align}
Similar to (\ref{eq_uu2}), we have
\begin{align}\label{eq_ad3}
|d_k(\pmb{\theta})| \mathop < d_{\text{max}} < \infty,
\end{align} for some $d_{max}$, and thus
\begin{align}\label{eq_ad1}
|n_k(\pmb{\theta}) + d_k(\pmb{\theta})| \mathop < d_{\text{max}} + n_{\text{max}} < \infty.
\end{align}
Then, from (\ref{eq_lc2}), (\ref{eq_lc3}), (\ref{eq_ad3}), (\ref{eq_ad1}) and by applying the product of functions in Section \ref{prod_lipz}, we can see that $\left(n_k(\pmb{\theta}) + d_k(\pmb{\theta})\right)d_k(\pmb{\theta})$ is Lipschitz continuous with a Lipschitz constant given by 
\begin{align}\label{eq_ad2}
L_{n_k+d_k,d_k} =  d_{\text{max}}L_{d_k} + (n_{\text{max}}+d_{\text{max}})L_{n_k,d_k}.
\end{align}
Next, we have
\begin{align}\label{eq_lbu1}
 \text{ln}2\left(n_k(\pmb{\theta}) + d_k(\pmb{\theta})\right)d_k(\pmb{\theta}) \ge \text{ln}2.
\end{align}
We now can conclude, from (\ref{eq_uu3}), (\ref{eq_uu4}), (\ref{eq_ad2}), and (\ref{eq_lbu1}) and  by applying the quotient of functions in the Section \ref{quoti_subsec}, that $\frac{n_k(\pmb{\theta}) \nabla d_k(\pmb{\theta})}{\text{ln}2\left(n_k(\pmb{\theta}) + d_k(\pmb{\theta})\right)d_k(\pmb{\theta})}$ is Lipschitz continuous with the following Lipschitz constant
\begin{align}\label{eq_lbu2}
 L_{t2}=\frac{n_{\text{max}}\lambda_{\text{max}}(\mathbf{C}_{d_k})\sqrt{M}L_{n_k,\nabla d_k}}{\sqrt{N}} + \frac{L_{n_k+ d_k,d_k}}{(\ln 2)^2}.
\end{align}
Finally, from (\ref{eq_lc4}), and (\ref{eq_lbu2}), by applying the linear combination of the functions in Section \ref{sec_lc}, it follows that $\nabla u_k(\pmb{\theta})$ is Lipschitz continuous with a Lipschitz constant given by
\begin{align}\label{eq_lbu3}
L_{\nabla u_k} =  L_{t1} +  L_{t2},
\end{align}
and $||\nabla u_k(\pmb{\theta})||$ is bounded as
\begin{align}\label{eq_gub}
	||\nabla u_k(\pmb{\theta})||
\le \frac{(\lambda_{\text{max}}(\mathbf{C}_{n_k}) + n_{\text{max}}\lambda_{\text{max}}(\mathbf{C}_{d_k}))\sqrt{M}}{\ln 2\sqrt{N}} \triangleq \zeta_{\nabla u_k}.
\end{align}

\subsubsection{  $u_k(\pmb{\theta})$ is bounded and Lipschitz continuous}
First, we have
\begin{equation}\label{uk_bounded}
u_k(\pmb{\theta}) =\frac{\tauc-\taup}{\tauc}\log_2\left(1+\frac{n_k(\pmb{\theta})}{d_k(\pmb{\theta})}\right) \le \frac{\tauc-\taup}{\tauc}\log_2\left(1 + n_{\text{max}} \right) \triangleq \zeta_{ u_k}. 
\end{equation}
Next, to prove  $u_k(\pmb{\theta})$ is Lipschitz continuous, we first consider the function
$h_1(x)=\text{log}_2(x)$  over the domain $1 \le x \le x_{\text{max}}$. 
 Note that $h_1(x)$ is continuously differentiable and thus we have
\begin{align}
\sup _{t \in(1, x_{\text{max}})}\left|h_{1}^{\prime}(t)\right|=\sup _{t \in(1, x_{\text{max}})}\log_2 {e}|\frac{1}{t}| \leq \log_2 {e}.
\end{align}
By the mean value theorem, there exists some $\xi$ between $y$ and $z,$  $y, z \in(1, x_{\text{max}})$, such that
\begin{align}\label{eq_lh1}
\left|h_{1}(z)-h_{1}(y)\right|=\left|h_{1}^{\prime}(\xi)(z-y)\right| \leq \sup _{t \in(1, x_{\text{max}})}\left|h_{1}^{\prime}(t)\right||z-y|\leq \text{log}_2 {e}|z-y|.
\end{align}
In other words, $h_1(x)=\text{log}_2(x)$ is Lipschitz continuous with a constant $ \text{log}_2 {e} $. 
Next we consider the function $h_2(x)= 1+\frac{n_k(\pmb{\theta})}{d_k(\pmb{\theta})} $, for which we have
\begin{align}\label{eq_db1}
1 \le |d_k(\pmb{\theta})| \mathop < \limits ^{(d3)} d_{\text{max}} < \infty
\end{align}
where $(d3)$ is based on the fact that $|d_k(\pmb{\theta})|$ is bounded by the power constraints. Then,  from (\ref{eq_lc1}), (\ref{eq_lc2}), (\ref{eq_uu2}), and (\ref{eq_db1}), the quotient property in Section \ref{quoti_subsec} implies that  $h_2(x)$ is  Lipschitz continuous with a Lipschitz constant found as
\begin{align}\label{eq_db2}
L_{h_2} = (n_{\text{max}}+d_{\text{max}})(L_n + L_d) + L_d.
\end{align}
Finally, from (\ref{eq_lh1}), and (\ref{eq_db2}) by applying composition property in Section \ref{subsec_compfunc}, we can prove that $u_k(\pmb{\theta})$ is  Lipschitz continuous with the following Lipschitz constant 
\begin{align}\label{eq_db3}
L_{u_k} = L_{h_2}\log_2 {e}.
\end{align}
\subsubsection{ $v(\pmb{\theta})$ and  $\nabla v(\pmb{\theta})$ are bounded and Lipschitz continuous}
It is easy to see that $v(\pmb{\theta})$ is bounded by
\begin{align}\label{eq_db4}
v(\pmb{\theta}) \le \bar{P}_{\mathrm{fix}}+\Pd N_{0}N\sum\limits _{m=1}^{M}\frac{1}{N\alpha_{m}} \triangleq \zeta_{v}.
\end{align}
Next, let us rewrite   $ v(\pmb{\theta}) $ as
\begin{align}\label{eq_v2}
v(\pmb{\theta}) =\bar{P}_{\mathrm{fix}}+\Pd N_{0}N\sum\limits _{m=1}^{M}\frac{1}{\alpha_{m}}||\pmb{\theta}_{m}||^{2} \triangleq \bar{P}_{\mathrm{fix}}+\mathbf{C}_{v}||\pmb{\theta}||^{2}.
\end{align}
Then, following the same steps in  Section \ref{sec_nuk}, we can show that $v(\pmb{\theta})$ is Lipschitz continuous with a Lipschitz constant $L_{v}$ expressed as
\begin{align}\label{eq_v3}
L_{v} = \frac{2\sqrt{M}\lambda_{\text{max}}(\mathbf{C}_{v})}{\sqrt{N}}.
\end{align}
Next, we have
\begin{align}\label{eq_v4}
\nabla v(\pmb{\theta}) =\Pd N_{0}N\left[\frac{2}{\alpha_{1}}\pmb{\theta}_{1};\frac{2}{\alpha_{2}}\pmb{\theta}_{2};\ldots,\frac{2}{\alpha_{M}}\pmb{\theta}_{M}\right] \triangleq \mathbf{C}_{\nabla v}\pmb{\theta},
\end{align}
and thus $||\nabla v(\pmb{\theta})||$ is bounded by
\begin{align}\label{eq_v5}
||\nabla v(\pmb{\theta})|| \le  \frac{\sqrt{M}\lambda_{\text{max}}(\mathbf{C}_{\nabla v})}{\sqrt{N}} \triangleq \zeta_{\nabla v},
\end{align}
and $\nabla v(\pmb{\theta})$ is Lipschitz continuous with the following Lipschitz constant
\begin{align}\label{eq_v6}
L_{\nabla v} = \lambda_{\text{max}}(\mathbf{C}_{\nabla v}).
\end{align}

\subsubsection{ $g_k(\pmb{\theta})$ and $\nabla g_k(\pmb{\theta})$ are bounded and Lipschitz continuous}\label{sec_gk} 
By following the same method in the Section \ref{sec_nuk} and using the fact that $d_k(\pmb{\theta}) \ge 1$, we can show that $g_k(\pmb{\theta})$  is  Lipschitz continuous with $L_{g_k} = a_{k} L_{d_k} + ||\tilde{\pmb{\gamma}}_{kk}^{T}\mathbf{A}_{k}||$ and bounded as
\begin{equation}
g_k(\pmb{\theta}) \le a_{k} d_k(\pmb{\theta}) + \tilde{\pmb{\gamma}}_{kk}^{T}\mathbf{A}_{k}\pmb{\theta}  \le  a_{k}|d_k(\pmb{\theta})| + \frac{n_{\text{max}}}{\Pd N^{2}}  <  a_{k}d_{\text{max}} + \frac{n_{\text{max}}}{\Pd N^{2}}  \triangleq \zeta_{g_k}.
\end{equation}
$\nabla g_k(\pmb{\theta})$ is also  Lipschitz continuous with $L_{\nabla g_k}=a_k L_{\nabla d_k}$ and bound as
\begin{equation}
	||\nabla g_k(\pmb{\theta})|| < \frac{a_{k}\Pd \lambda_{\text{max}}(\mathbf{C}_{d_k})\sqrt{M}}{\sqrt{N}} + \lambda_{\text{max}}(\mathbf{A}_{k}^{T}\tilde{\pmb{\gamma}}_{kk}) \triangleq \zeta_{\nabla g_k}.
\end{equation}

\subsubsection{ $\nabla f_{\xi}(\pmb{\theta})$ is Lipschitz continuous} 
From (\ref{eq:rate:rewrite2}) we have \begin{equation}
v(\pmb{\theta})^2 \ge \bar{P}^2_{\mathrm{fix}},
\end{equation}
which is due to the fact that the  second  term  in  (\ref{eq:rate:rewrite2})  is  always greater than or equal to zero. Then, following the same method to find the Lipschitz constant of $\nabla u_k(\theta)$ in  Section \ref{sec_nuk}, we can prove that $\nabla f_{\xi}(\pmb{\theta})$ is Lipschitz continuous with a Lipschitz constant given by 
\begin{align}\label{eq_lf1}
L_{f_{\xi}} = L_{1} + L_{2} + L_{3},
\end{align}
where 
$L_{1} = \frac{B\sum_{k=1}^{K} L_{\nabla u_{k}}}{\bar{P}^2_{\mathrm{fix}}}$; $L_2 = \frac{B\left((\sum_{k=1}^{K} L_{\nabla u_{k}})(\sum_{k=1}^{K} \zeta_{\nabla u_{k}}) + L_{v}\zeta_{v}\right)}{\bar{P}^4_{\mathrm{fix}}} $; and $L_{3} = 2\xi \sum_{k=1}^{K} (\zeta_{g_k}L_{g_k} + \zeta_{\nabla g_k}L_{\nabla g_k})$.

\subsection{Proof of Proposition \ref{lem:grad}}\label{sec:grad_proof}

Recall that the PF is
\begin{equation}
\Psi_{k} (\pmb{\theta}) = \left[\text{max}\left(0,  { g}_{k}(\pmb{\theta})\right )\right]^2,
\end{equation}
then
\begin{equation}
\nabla \Psi_{k} (\pmb{\theta}) = \begin{cases}
0, &  g_{k}(\pmb{\theta}) \le 0 \\
2 { g}_{k}(\pmb{\theta})\nabla g_{k}(\pmb{\theta}), &  g_{k}(\pmb{\theta}) >  0,
\end{cases}
\end{equation}
which can be rewritten as
\begin{equation}
\nabla \Psi_{k} (\pmb{\theta}) = 2  \left[\text{max}\left(0, { g}_{k}(\pmb{\theta}) \right )\right] \nabla g_{k}(\pmb{\theta}).
\end{equation}
Using the quotient rule we can write $\nabla f(\pmb{\theta})$ as
\begin{align}
\nabla f_{\xi}(\pmb{\theta})= B\frac{v(\pmb{\theta})\nabla u(\pmb{\theta})-u(\pmb{\theta})\nabla v(\pmb{\theta})}{v(\pmb{\theta})^{2}} - \xi\sum_{k=1}^{K} \nabla \Psi_{k} (\pmb{\theta}),
\end{align}
where 
\begin{equation}
\nabla v(\pmb{\theta})=\Pd N_{0}N\bigl[\frac{2}{\alpha_{1}}\pmb{\theta}_{1};\frac{2}{\alpha_{2}}\pmb{\theta}_{2};\ldots,\frac{2}{\alpha_{M}}\pmb{\theta}_{M}\bigr],
\end{equation}
and 
\begin{equation}
\nabla u(\pmb{\theta})= \sum_{k=1}^{K} \nabla u_{k}(\pmb{\theta}).
\end{equation}
To find the gradient of $\nabla g_k(\pmb{\theta})$ we recall the following
equalities
\begin{equation}
\nabla\bigl(\tilde{\pmb{\gamma}}_{k'k}^{T}\mathbf{A}_{k'}\pmb{\theta}\bigr)^{2}=2\mathbf{A}_{k'}^{T}\tilde{\pmb{\gamma}}_{k'k}\tilde{\pmb{\gamma}}_{k'k}^{T}\mathbf{A}_{k'}\pmb{\theta},
\end{equation}
\begin{equation}
\nabla\bigl(||\tilde{\pmb{\kappa}}_{k}\odot(\mathbf{A}_{k'}\pmb{\theta})||^{2}\bigr)=2\mathbf{A}_{k'}^{T}\mathbf{B}_{k}\mathbf{A}_{k'}\pmb{\theta},
\end{equation}
where $\mathbf{B}_{k}\in R_{+}^{M\times M}$ is a diagonal matrix
whose $m$-th element is $[\mathbf{B}_{k}]_{m}=\beta_{mk}$. Then, by applying the chain rule, we can easily compute the gradient of $\nabla g_k(\pmb{\theta})$ as shown in (\ref{eq_gk}). To find the gradient of $\nabla u_k(\pmb{\theta})$, we first apply the chain rule together with the quotient rule, we have
\begin{align}\label{eq:grad_uk1}
\nabla u_k(\pmb{\theta}) &= \frac{d(\pmb{\theta})}{\text{ln}2\left(n(\pmb{\theta}) + d(\pmb{\theta})\right)} \frac{d(\pmb{\theta}) \nabla n(\pmb{\theta}) - n(\pmb{\theta}) \nabla d(\pmb{\theta})}{d(\pmb{\theta})^2} \notag\\
&= \frac{1}{\text{ln}2\left(n(\pmb{\theta}) + d(\pmb{\theta})\right)}\left(\nabla n(\pmb{\theta}) - \frac{n(\pmb{\theta}) \nabla d(\pmb{\theta})}{d(\pmb{\theta})}\right),
\end{align}

where
\begin{equation}
n(\pmb{\theta}) \triangleq \Pd N^{2}\bigl(\tilde{\pmb{\gamma}}_{kk}^{T}\mathbf{A}_{k}\pmb{\theta}\bigr)^{2},
\end{equation}

\begin{equation}
d(\pmb{\theta}) \triangleq \Pd N^{2}\sum\limits _{k'\neq k}^{K}\bigl(\tilde{\pmb{\gamma}}_{k'k}^{T}\mathbf{A}_{k'}\pmb{\theta}\bigr)^{2}+\Pd N\sum\limits _{k'=1}^{K}||\tilde{\pmb{\kappa}}_{k}\odot(\mathbf{A}_{k'}\pmb{\theta})||^{2}+1,
\end{equation}

\begin{align}\label{eq:f1}
\nabla n(\pmb{\theta}) = 2\Pd N^{2}\mathbf{A}_{k}^{T}\tilde{\pmb{\gamma}}_{kk}\tilde{\pmb{\gamma}}_{kk}^{T}\mathbf{A}_{k}\pmb{\theta},
\end{align}
and 
\begin{align}\label{eq:f2}
\nabla d(\pmb{\theta}) = 2\Pd N^{2}\sum_{k'\ne k}^{K}\mathbf{A}_{k'}^{T}\tilde{\pmb{\gamma}}_{k'k}\tilde{\pmb{\gamma}}_{k'k}^{T}\mathbf{A}_{k'}\pmb{\theta} + 2\Pd N\sum_{k'=1}^{K}\mathbf{A}_{k'}^{T}\mathbf{B}_{k}\mathbf{A}_{k'}\pmb{\theta}.
\end{align}
The substitution of (\ref{eq:f1}) and (\ref{eq:f2}) into (\ref{eq:grad_uk1}) yields (\ref{eq:grad:SE}).

\subsection{Proof of Proposition \ref{lem:proj}}\label{sec_proj}

We now show that the projection onto $\mathcal{C}$ admits an \emph{analytical
	solution} and is \emph{parallelizable}. Recall that $P_{\mathcal{C}}(\mathbf{u})$
is explicitly written as \begin{subequations}\label{eq:proj:main}
	\begin{align}
	\underset{\pmb{\theta}\in\mathbb{R}^{MK}}{\min} & \quad||\pmb{\theta}-\mathbf{u}||^{2}\\
	\st & \quad||\pmb{\theta}_{m}||^{2}\leq\frac{1}{N},m=1,2,\ldots,M\\
	& \quad\pmb{\theta}\geq 0.
	\end{align}
\end{subequations} Note that the objective in \eqref{eq:proj:main} is
separable with $\pmb{\theta}_{m}$. Thus \eqref{eq:proj:main} boils
down to solving the following subproblem for each $m$ \begin{subequations}\label{eq:proj:sub}
	\begin{align}
	\underset{\pmb{\theta}_{m}\in\mathbb{R}^{K}}{\min} & \quad||\pmb{\theta}_{m}-\mathbf{u}_{m}||^{2}\\
	\st & \quad||\pmb{\theta}_{m}||^{2}\leq\frac{1}{N}\\
	& \quad\pmb{\theta}_{m}\geq 0.
	\end{align}
\end{subequations} Problem \eqref{eq:proj:sub} is actually the projection
onto the intersection of an Euclidean ball and the positive orthant. Finally, the result (\ref{eq:anal_sol}) is a direct application of \cite[Theorem 7.1]{Bauschke:2018:Proj}.

\subsection{Convergence Proof of Algorithm \ref{alg:APG}}\label{Appe1}
The convergence proof is divided into two parts. In the first part of the proof we show that, for a given $\xi_{m}$, the APG iterations converge to a stationary solution to the penalized problem \eqref{eq:opt_theta}. In the second part we show that $\{\pmb{\theta}_{\xi}\}$ converges to a feasible point of $(\mathcal{P}_2)$. Thus the convergent point of Algorithm  \ref{alg:APG} is indeed a stationary point of $(\mathcal{P}_2)$.

We begin with the first part of the proof  by recalling an important inequality of a $L_{f}$-Lipschitz continuous gradient
function. Specifically, for a function $f(\mathbf{x})$ has a Lipschitz continuous
gradient with a constant $L_f$, the following inequality holds
\begin{equation}
f(\mathbf{y})\geq f(\mathbf{x})+\bigl \langle{\nabla_{\mathbf{x}}}f(\mathbf{x}),\mathbf{y}-\mathbf{x}\bigr\rangle-\frac{L_{f}}{2}||\mathbf{y}-\mathbf{x}||^{2}.\label{eq:Lips:grad:inequality}
\end{equation}
The projection in \eqref{eq:updatev} is equivalent to \begin{subequations}
\begin{align}\label{eq_cv0}
\mathbf{v}^{(n+1)}&=\underset{\pmb{\theta}\in \mathcal{C}}{\operatorname{argmin}}\ || \pmb{\theta}-\pmb{\theta}^{(n)} -\alpha_{\theta} \nabla f_{\xi}(\pmb{\theta}^{(n)} ||^2\\
&=\underset{\pmb{\theta}\in \mathcal{C}}{\operatorname{argmax}}\ \langle \nabla f_{\xi}(\pmb{\theta}^{(n)},\pmb{\theta}-\pmb{\theta}^{(n)}\rangle - \frac{1}{2\alpha_{\theta}}||\pmb{\theta}-\pmb{\theta}^{(n)}||^2 \label{eq:updatev:rewrite}
\end{align}
\end{subequations}
where $\langle \mathbf{x},\mathbf{y}\rangle=\mathbf{x}^{T}\mathbf{y}$ is the inner product of $ \mathbf{x} $ and $ \mathbf{y} $ and we have used the fact that $||\mathbf{a}-\mathbf{b}||^{2}=||\mathbf{a}||^{2}+||\mathbf{b}||^{2}-2\langle \mathbf{a},\mathbf{b}\rangle$. Note that when $\pmb{\theta}=\pmb{\theta}^{(n)}$, the objective in \eqref{eq:updatev:rewrite} is $ 0 $ and $\mathbf{v}^{(n+1)}$  is the optimal solution to  \eqref{eq:updatev:rewrite}. Thus the following inequality is obvious
\begin{equation}
\langle \nabla f_{\xi}(\pmb{\theta}^{(n)},\mathbf{v}^{(n+1)}-\pmb{\theta}^{(n)}\rangle - \frac{1}{2\alpha_{\theta}}||\mathbf{v}^{(n+1)}-\pmb{\theta}^{(n)}||^2 \geq 0 \label{eq:update:ineqn}
\end{equation}
Combining \eqref{eq:Lips:grad:inequality} and \eqref{eq:update:ineqn}  we obtain
\begin{align}\label{eq_cv1}
f_{\xi}\big(\mathbf{v}^{(n+1)}\big) &\geq f_{\xi}\big(\pmb{\theta}^{(n)}\big)+\big\langle\nabla f_{\xi}\big(\pmb{\theta}^{(n)}\big), \mathbf{v}^{(n+1)}-\pmb{\theta}^{(n)}\big\rangle  -\frac{L}{2}\big\|\mathbf{v}^{(n+1)}-\pmb{\theta}^{(n)}\big\|^{2} \nonumber\\
&\geq f_{\xi}\big(\pmb{\theta}^{(n)}\big)+\Big(\frac{1}{2 \alpha_{\theta}}-\frac{L}{2}\Big)\big\|\mathbf{v}^{(n+1)}-\pmb{\theta}^{(n)}\big\|^{2}.
\end{align}
It is easy to see that if $\alpha_{\theta}<\frac{1}{L_{f}}$, then $f_{\xi}\big(\mathbf{v}^{(n+1)}\big)\geq f_{\xi}\big(\pmb{\theta}^{(n)}\big)$.
Next, if $f_{\xi}\big(\mathbf{z}^{(n+1)}\big) \geq f_{\xi}\big(\mathbf{v}^{(n+1)}\big)$, then $\pmb{\theta}^{(n+1)}=\mathbf{z}^{(n+1)}$, and
\begin{align}\label{eq_cv2}
f_{\xi}\big(\pmb{\theta}^{(n+1)}\big)=f_{\xi}\big(\mathbf{z}^{(n+1)}\big) \geq f_{\xi}\big(\mathbf{v}_{k+1}\big).
\end{align}
If $f_{\xi}\big(\mathbf{z}^{(n+1)}\big) < f_{\xi}\big(\mathbf{v}^{(n+1)}\big)$, then $\pmb{\theta}^{(n+1)}=\mathbf{v}^{(n+1)}$, and
\begin{align}\label{eq_cv3} 
f_{\xi}\big(\pmb{\theta}^{(n+1)}\big)=f_{\xi}\big(\mathbf{v}^{(n+1)}\big)\geq f_{\xi}\big(\pmb{\theta}^{(n)}\big).
\end{align}
In summary we have shown that
\begin{align}
f_{\xi}\big(\pmb{\theta}^{(n+1)}\big) \geq f_{\xi}\big(\mathbf{v}^{(n+1)}\big) \geq  f_{\xi}\big(\pmb{\theta}^{(n)}\big).
\end{align}
Since $ \mathcal{C} $ is compact convex, $\{\pmb{\theta}^{(n)}\}$ and  $\{\mathbf{v}^{(n)}\}$ are bounded. Thus  $\{\pmb{\theta}^{(n)}\}$ has accumulation
points.

In the second part of the proof, we now assert that any accumulation point is a stationary solution of \eqref{eq:opt_theta}.
 As $f_{\xi}(\pmb{\theta}^{(n)})$ is non-decreasing, the objective at all the accumulation points is the same which is denoted by $ f_{\xi}^{\ast}$. Then, from (\ref{eq_cv1}) we have
\begin{align}
\Big(\frac{1}{2 \alpha_{\theta}}-\frac{L}{2}\Big)\Big\|\mathbf{v}^{(n+1)}-\pmb{\theta}^{(n)}\Big\|^{2} \leq f_{\xi}\big(\mathbf{v}^{(n+1)}\big)-f_{\xi}\big(\pmb{\theta}^{(n)}\big) \leq f_{\xi}\big(\pmb{\theta}^{(n+1)}\big) -f_{\xi}\big(\pmb{\theta}^{(n)}\big) .
\end{align}
Summing over $n=1,2, \cdots, \infty$, we have
\begin{align}
\Big(\frac{1}{2 \alpha_{\theta}}-\frac{L}{2}\Big) \sum_{n=1}^{\infty}\big\|\mathbf{v}^{(n+1)}-\pmb{\theta}^{(n)}\big\|^{2} \leq f_{\xi}^{\ast}-f_{\xi}\big(\pmb{\theta}^{(1)}\big)<\infty.
\end{align}
Since  $\alpha_{\mathbf{x}}<\frac{1}{L}$, we can conclude that 
\begin{align} \label{eq_cv_a}
\mathbf{v}^{(n+1)} \rightarrow \pmb{\theta}^{(n)} \quad \text { as } \quad n \rightarrow \infty. 
\end{align}
The optimality condition of (\ref{eq:updatev:rewrite}) results in
\begin{align}
\big\langle \frac{1}{\alpha_{\theta}}\big(\mathbf{v}^{(n+1)}-\pmb{\theta}^{(n)}\big)-\nabla_{\pmb{\theta}}f_{\xi}(\pmb{\theta}^{(n)} ),\pmb{\theta} -\pmb{\theta}^{(n)} \big\rangle &\leq 0, \quad\forall \pmb{\theta} \in \mathcal{C}. \label{eq:optcond}
\end{align}
Let $ \pmb{\theta}^{\ast}$  be any accumulation point of $\{\pmb{\theta}^{(n)}\}$, i.e. $\{\pmb{\theta}^{(n_j)}\} \rightarrow \pmb{\theta}^{\ast}$ as $j \rightarrow \infty$. From \eqref{eq_cv_a} we immediately have that $  \{\pmb{\theta}^{(n_{j}+1)}\} \rightarrow \pmb{\theta}^{\ast} $. We also note that $\nabla_{\pmb{\theta}} f_{\xi}(\pmb{\theta})$ is continuous, and thus $ \nabla_{\pmb{\theta}} f_{\xi}(\pmb{\theta}^{(n_j)}) \rightarrow \nabla_{\pmb{\theta}} f_{\xi}(\pmb{\theta}^{\ast}) $. By letting $ j\to \infty $ in \eqref{eq:optcond}, we have
$\big\langle \nabla_{\pmb{\theta}}f_{\xi}(\pmb{\theta}^{\ast} ),\pmb{\theta} -\pmb{\theta}^{\ast} \big\rangle \geq 0, ~ \forall \pmb{\theta} \in \mathcal{C}.$ 
This inequality simply means that $ \pmb{\theta}^{\ast}  $ is a stationary solution to \eqref{eq:opt_theta} which completes the first part of the proof.

Now we show that $\{\pmb{\theta}_{\xi_{m}}\}$ indeed converges to a feasible point of $(\mathcal{P}_{2})$. Note that for small
$\xi_{m}$, $\pmb{\theta}_{\xi_{m}}$ may not be feasible to $(\mathcal{P}_{2})$,
and that the following inequalities always hold for lager $\xi_{m}$
\begin{equation}
f(\boldsymbol{\theta}_{\xi_{m}})\geq f_{\xi_{m}}(\boldsymbol{\theta}_{\xi_{m}})=\underset{\boldsymbol{\theta}\in\mathcal{C}}{\max}\ f_{\xi_{m}}(\boldsymbol{\theta})\geq f^{\ast},\label{eq:1stinequality}
\end{equation}
where $f^{\ast}$ is the optimal objective of $(\mathcal{P}_{2})$.
In the above, the first equality is due to the negativity of the penalty
term and the second inequality is true for two reasons. First, the problem $\underset{\boldsymbol{\theta}\in\mathcal{C}}{\max}\ f_{\xi_{m}}(\boldsymbol{\theta})$ becomes a convex problem for  large $\xi_{m}$ since $f_{\xi_{m}}(\boldsymbol{\theta})$ becomes concave. Second, the APG method can find the optimal solution since the problem is now convex. Thus the second inequality in \eqref{eq:1stinequality} holds because the optimal objective is no less than the objective at any feasible solution. 

Let us consider a sequence $\xi_{m}\to\infty$. Since the sequence
$\boldsymbol{\theta}_{\xi_{m}}$ belongs to compact set, it has a
convergent subsequence (i.e. the Bolzano-Weierstrass theorem). Thus
we can assume without loss of optimality that $\boldsymbol{\theta}_{\xi_{m}}$
converges to a certain point $\boldsymbol{\theta}^{\ast}$ by abuse of notation. We will
show that $\boldsymbol{\theta}^{\ast}$ 
is indeed 
feasible to $(\mathcal{P}_{2})$. First note that since
$\boldsymbol{\theta}_{\xi_{m}}\to\boldsymbol{\theta}^{\ast}$ and
thus $f(\boldsymbol{\theta}_{\xi_{m}})\to f(\boldsymbol{\theta}^{\ast})$
due to the continuity of $f$. Then from (\ref{eq:1stinequality})
we have 
$f(\boldsymbol{\theta}^{\ast})\geq f^{\ast}.$
Suppose to the contrary that $\boldsymbol{\theta}^{\ast}$ is infeasible.
Since $\boldsymbol{\theta}_{\xi_{m}}\to\boldsymbol{\theta}^{\ast}$
and $\Psi_{k}(\cdot)$ is continuous, for sufficiently large $m$
we have 
\begin{equation}
\Psi_{k}(\boldsymbol{\theta}_{\xi_{m}})\geq\Psi_{k}(\boldsymbol{\theta}^{\ast})>0,k=1,2,\ldots,K.
\end{equation}
Thus for these $k$ we would have 
\begin{equation}
f_{\xi_{m}}(\boldsymbol{\theta}_{\xi_{m}})=f(\boldsymbol{\theta}_{\xi_{m}})-\xi_{m}\Psi_{k}(\boldsymbol{\theta}_{\xi_{m}})\leq f(\boldsymbol{\theta}_{\xi_{m}})-\xi_{m}\Psi_{k}(\boldsymbol{\theta}^{\ast}).
\end{equation}
Thus it is easy to see that $f_{\xi_{m}}(\boldsymbol{\theta}_{\xi_{m}})\to-\infty$
when $\xi_{m}\to\infty$, which contradicts (\ref{eq:1stinequality})
and thus completes the proof.

\end{document}